\documentclass[aps,10pt,pra,twocolumn,groupedaddress,
  floatfix, superscriptaddress,showpacs,showkeys,footinbib]{revtex4-2}

\usepackage{amsfonts}
\usepackage{amsmath}
\usepackage{amsthm}
\newtheorem{theorem}{Theorem}

\usepackage{amssymb, setspace}
\usepackage{graphicx}
\usepackage{multirow}
\usepackage[en-GB]{datetime2}
\DTMlangsetup{showdayofmonth=false}

\usepackage{umoline}
\usepackage[usenames,svgnames]{xcolor}

\newcommand{\cRef}[1]{Ref.~\cite{#1}}
\newcommand{\cc}[1]{~\cite{#1}}
\newcommand{\av}[1]{{ \langle {#1} \rangle }}
\DeclareMathOperator*{\Tr}{Tr}
\newcommand{\orderof}[1]{O(#1)}
\newcommand{\EqDef}{\stackrel{\mathrm{def}}{=}}
\newcommand{\Id}{I}

\newcommand{\eq}[1]{Eq.~\eqref{#1}}
\newcommand{\Eq}[1]{Eq.~(\ref{#1})}
\newcommand{\Eqs}[2]{Eqs.~(\ref{#1},~\ref{#2})}

\newcommand{\App}[1]{Appendix~\ref{#1}}
\newcommand{\Table}[1]{Table~\ref{#1}}

\newcommand{\Fig}[1]{Fig.~\ref{#1}}

\newcommand{\Figss}[3]{Figs.~\ref{#1},~\ref{#2},~\ref{#3}}
\newcommand{\Sec}[1]{Sec.~\ref{#1}}
\global\long\def\norm#1{\left\Vert #1\right\Vert }

\global\long\def\ket#1{\left|#1\right\rangle }%
\global\long\def\bra#1{\left\langle #1\right|}%
\global\long\def\braket#1#2{\left\langle #1\right.\left|#2\right\rangle }%

\newcommand{\Bv}{\bm{v}}

\usepackage[
    colorlinks=true,
    urlcolor=blue,
    linkcolor=red,
    citecolor=DarkGreen]{hyperref}

\usepackage[margin=2cm]{geometry}
\usepackage[all]{hypcap}
\usepackage{color}
\usepackage{bm}
\usepackage{nicefrac}

\usepackage{titlesec}
\titleformat{\subsection}[runin]{\itshape\bfseries}{\thesubsection.}{1em}{}[.]
\titleformat{\subsubsection}[runin]{\itshape}{\thesubsubsection.}{1em}{}[.]

\usepackage{booktabs}
\usepackage{cleveref}

\begin{document}
\date{\today}
\author{Yotam Y. Lifshitz}
\email[Yotam Y. Lifshitz: ]{yotamlif@campus.technion.ac.il}
\affiliation{Physics Department, Technion, 3200003, Haifa, Israel}
\affiliation{IBM Quantum, IBM Research -- Haifa, Haifa University Campus, Mount Carmel, Haifa 31905, Israel}
\author{Eyal Bairey}
\affiliation{Physics Department, Technion, 3200003, Haifa, Israel}
\author{Eli Arbel}
\affiliation{IBM Quantum, IBM Research -- Haifa, Haifa University Campus, Mount Carmel, Haifa 31905, Israel}
\author{Gadi Aleksandrowicz}
\affiliation{IBM Quantum, IBM Research -- Haifa, Haifa University Campus, Mount Carmel, Haifa 31905, Israel}
\author{Haggai Landa}
\affiliation{IBM Quantum, IBM Research -- Haifa, Haifa University Campus, Mount Carmel, Haifa 31905, Israel}
\author{Itai Arad}
\affiliation{Physics Department, Technion, 3200003, Haifa, Israel}
\title{Practical Quantum State Tomography for Gibbs states}
\begin{abstract}
    Quantum state tomography is an essential tool for the
    characterization and verification of quantum states. However, as
    it cannot be directly applied to systems with more than a few
    qubits, efficient tomography of larger states on
    mid-sized quantum devices remains an important challenge in
    quantum computing. We develop a tomography approach that
    requires moderate computational and quantum resources for the
    tomography of states that can be approximated by Gibbs
    states of local Hamiltonians. The proposed method, Hamiltonian
    Learning Tomography, uses a Hamiltonian learning algorithm to
    get a parametrized ansatz for the Gibbs Hamiltonian, and
    optimizes it with respect to the results of local measurements.
    We demonstrate the utility of this method with a high fidelity
    reconstruction of the density matrix of 4 to 10 qubits in a
    Gibbs state of the transverse-field Ising model, in numerical
    simulations as well as in experiments on IBM Quantum
    superconducting devices accessed via the cloud. Code implementation
    of the our method is freely available as an open source software in
    Python.
\end{abstract}
\maketitle

\section{Introduction}\label{sec:introduction}

Quantum State Tomography (QST) is the process of reconstructing the
density matrix of a quantum state by measuring multiple copies of
it\cc{PhysRevA.40.2847}. QST is a useful tool for the
characterization and verification of quantum hardware, and can also
be used for understanding fundamental properties of a system such as
entanglement\cc{White1999,Leonhardt1995}. There exist several
well-established QST methods, which have been implemented on many
types of devices and setups\cc{Wang_2016, Thew_2002,
Lee_2002,Resch_2005,QiskitQST}. However, the number of measurements
required in common QST methods, such as maximum likelihood
estimation, grows exponentially with the number of qubits in the
system\cc{QuantumStateEstimation, SampleOptimalTomography},
rendering these methods infeasible for systems with more than a few
qubits. There are several approaches to reduce the number of
measurements needed for QST, either by making assumptions on the
quantum state
\cc{Torlai2018,Gross2010,Hffner2005, Titchener_2018, cramer2010efficient},
by circumventing the need to reconstruct the full density
matrix\cc{aaronson2018shadow,Huang_2020} or by making adaptive
measurement process\cc{rambach2021robust}.

An important class of quantum states in condensed matter physics and
quantum computing are Gibbs states of local Hamiltonians.
These are states $\rho$ that can be parametrized as
\begin{align}
  \rho = \frac{1}{Z}e^{-H}
  \label{Eq:Gibbs0},
\end{align}
where $Z$ is an overall
normalization factor, and $H$, called the Gibbs Hamiltonian (GH), is
(at least approximately) a $k$-local operator; namely, $H$ is a sum
of terms that act non-trivially on $k$ spatially contiguous degrees
of freedom according to the connectivity topology of some relevant
system. Gibbs states describe the equilibrium properties of
many-body quantum systems\cc{vonNeumann}, and as such are expected
to approximately describe the dynamics of systems that thermalize
due to a coupling with a bath\cc{reichental}, as well as subsystems
of closed systems that thermalize under their own
dynamics\cc{ETHoriginal, ETHreview}. They are also used as a
model (quantum Boltzmann machine) in quantum machine
learning\cc{Amin_2018}, and their preparation is an essential part
of various quantum algorithms\cc{QuantumSimulatedAnnealing,
QuantumSDP}. Thus, states that can be approximately described as
Gibbs states ubiquitously arise in quantum simulations of
many-body systems and when quantum algorithms are implemented
on noisy devices.

In this work we introduce a practical tomography approach for
approximate Gibbs states that we call \emph{Hamiltonian Learning
Tomography} (HLT). Our main idea is to use the Hamiltonian learning
algorithm of \cRef{Bairey_2019} to efficiently identify a restricted
ansatz for the Gibbs Hamiltonian, described by a small number of
relevant parameters. Then, following the proposal of Entanglement
Hamiltonian Tomography (EHT) from \cRef{Kokail_2021}, and employing
improvements important in our context, the ansatz is optimized with
respect to the results of local measurements. Adapting elements from
these two methods allows us to significantly extend the scopes in
which they can be used and overcome drawbacks in each of them. All
of this is achieved while using moderate computational and quantum
resources.

As we will show, the use of the Hamiltonian learning algorithm to
generate an ansatz removes the need for an a-priori knowledge of the
Hamiltonian that parametrizes the state, as relied upon in EHT.
Moreover, the Hamiltonian learning ansatz turns out to require a
very low number of measurements for reaching a high tomography
precision, in comparison with traditional QST. At the same time,
using an optimization similar to that of EHT fixes the main
drawbacks of the Hamiltonian learning procedure of
\cRef{Bairey_2019}, as we in explain in \Sec{sec:ehlt}.

To demonstrate the potential of our approach as a practical tool for
characterizing the state of mid-sized quantum devices, we study it
in detail using both numerical simulations and experiments on actual
quantum hardware. We obtain a high fidelity of reconstruction of the
density matrix of 4 to 10 qubits on an IBM Quantum superconducting
device accessed via the cloud\cc{IBMQ}. Using a variational
algorithm to create the Gibbs state of a 1D transverse-field Ising
Hamiltonian, we observe fidelities of over $0.99$ in a system of 5
qubits using only fifty thousand measurements, while a traditional
QST reaches in our experiments fidelities of less than $0.6$ using
the same number of measurements (and about five million measurements
are required for a reliable result). For more than 5 qubits, the
density matrix cannot be obtained at all with brute-force methods.
Therefore, as an independent verification of the density matrix
reconstructed using our tomography approach for 6 to 10 qubits, we
use QST to directly measure subsystem reduced density matrices of 3
qubits, for all of which we observe a fidelity of over $0.99$ with
the corresponding subsystem states of the full HLT density matrix.
Finally, Greenberger-Horne-Zeilinger (GHZ)\cc{Greenberger1989}
states (also known as cat states) on quantum hardware are
presented as a test case where we can detect performance issues with
our method.

An open source Python code implementation of our method is freely available
for the community on Github\cc{Hltgithub}.

The structure of this paper is a follows; We start by giving a
general review of the previous methods in \Sec{sec:Background}. In
\Sec{sec:ehlt} we describe our method, and the states that fit into
it. Then, results from simulations and quantum hardware
demonstrating the method, are presented in \Sec{sec:results}. We
conclude with a discussion and outlook in
\Sec{sec:discussion-and-outlook}.

\section{Background}\label{sec:Background}

\subsection{Entanglement Hamiltonian Tomography}
\label{subsec:entanglement-hamiltonian-tomography}

Recently, in Refs.\cite{dalmonte2018quantum, Kokail_2021, kokail2021quantum}, a method called Entanglement
Hamiltonian Tomography was suggested for the tomography of states generated by
local Hamiltonians. Given a certain region in a many-body quantum system, its
reduced density matrix encodes the results of all possible measurements in
that region.  Instead of learning this density matrix directly, one
opts to learn the entanglement Hamiltonian, sometimes called modular
Hamiltonian, defined by:
\begin{align}
\label{eq:EH definition}
    H_E\EqDef-\ln\left( \rho \right),
\end{align}
where $\rho$ is the reduced density matrix, and hence \eq{eq:EH
definition} takes formally the same expression as \eq{Eq:Gibbs0}.

There are certain rigorous results, based on conformal field
theory{\cc{Calabrese_2009}} and the Bisognano Wichmann
theorem{\cc{Bisognano1975, Bisognano1976, Giudici2018}}, that guarantee the locality of the
entanglement Hamiltonian in various settings.
These apply for ground states and for
quantum quenches to a critical point of local Hamiltonians. The EHT
method heuristically builds on these results, which formally apply
for continuous, infinite and Lorentz invariant systems, and applies
them to discrete lattice systems and to time evolution. It suggests
an ansatz, based on these results, for the entanglement Hamiltonian,
which is a linear combination of the terms of the underlying
dynamical Hamiltonian. Given a Hamiltonian that is a sum of local
terms $H=\sum_i h_i$, the basic ansatz consists of the same terms,
but with arbitrary variational coefficients:
\begin{align*}
  H_E &=\sum_i \theta_i h_i,&  \theta_i &\in \mathbb{R}.
\end{align*}
The above ansatz can be further extended using physically motivated
terms (such as momentum terms). The vector of parameters
$\bm{\theta}$ is then optimized by sampling $\rho$ in several bases
of rotated product states, which are defined by random local
unitaries that form a $2$-design. The statistics of these
measurements is then used to find $\bm{\theta}$ using a cost
function and gradient based optimization methods (see
\Sec{subsec:optimization}).

\subsection{Hamiltonian learning using a constraint matrix}
\label{subsec:constraint-matrix}

Quantum Hamiltonian learning is the problem of
reconstructing a Hamiltonian by measuring its
dynamics\cc{Granade_2012,Wiebe_2014april,Wiebe_2014may,Wang_2017,
da2011practical,li2020hamiltonian,wang2015hamiltonian,zhang2014quantum,
samach2021lindblad} or a steady state of the dynamics
\cc{haah2021optimal,anshu2021sample,Rudinger_2015,
kieferova2017tomography,kappen2020learning,evans2019scalable}.
A recent method for learning local Hamiltonians from steady states
using local measurements\cc{Bairey_2019} can be applied for the
learning of Gibbs \emph{states} of local Hamiltonians. Given many
copies of a Gibbs state $\rho = \frac{1}{Z}e^{-H}$, we aim at
learning the Gibbs Hamiltonian $H$, instead of learning $\rho$
directly. The method begins by noting that $[H,\rho]=0$ and
therefore for \emph{every} operator $A$,
\begin{align}
\label{eq:Ehrenfest}
  \av{i[A,H]}_\rho \EqDef i\Tr([A,H]\rho) =
  i\Tr(A[H, \rho]) = 0,
\end{align}
where we used the cyclicity of the trace. Assuming $H$ is a
\emph{$k$-local} Hamiltonian, we first expand it in terms of a basis
of $k$-local operators $\{S_m\}_{m=1}^M$, which we conveniently
choose to be $k$-local Pauli operators (since the normalization
factor $Z$ of \eq{Eq:Gibbs0} can absorb the identity component of
$H$, we can assume without loss of generality that $H$ is
traceless). Then
\begin{align}
\label{eq:H_def}
  H = \sum_{m=1}^M v_m S_m,
\end{align}
and our task is to estimate the coefficients $\{v_m\}$. Choosing a
set of local ``constraint operators'' $\{A_q\}_{q=1}^Q$, where $Q\ge
M$, we use \Eq{eq:Ehrenfest} to obtain a linear constraint on
$\{v_m\}$ for each $A_q$;
\begin{align}
\label{eq:commutator}
    \av{i[A_q,H]}_\rho = \sum_m \av{i[A_q,S_m]}_\rho v_m = 0.
\end{align}
This set of equations can be
compactly written as a matrix equation for the vector $\bm{v} \EqDef
\{v_m\}$,
\begin{align}
    \label{eq:K_nm}
    K\bm{v}=0,\qquad K_{q,m}\equiv \left< i\left[ A_q,S_m \right]
    \right>_\rho,
\end{align}
where $K$ is a $Q\times M$ real matrix, which we call the
\emph{constraint matrix} (CM).  As it is a homogenous equation,
solving it gives $\bm{v}$ \emph{up to an overall factor}.  In order to have
more equations than unknown (i.e., $Q>M$), we take the set of
constraint operators to be all $(k+1)$-local Paulis.

The CM $K$ can be measured directly, as each of its entries is
an expectation value of at most a $2k$-local Pauli. However, any
finite number of measurement will necessarily produce statistical
errors in the measured CM:
\begin{align}
\label{eq:noisy-CM}
  \tilde{K} &= K+ \epsilon E,  &
  \epsilon &= \orderof{1/\sqrt{m}}.
\end{align}
Above, $\epsilon E$ is the statistical noise from the
measurements, and we have distinguished for clarity a formal
parameter $\epsilon$ that models the standard deviation of
statistical noise, which depends on the number of measurements $m$.
Consequently, the equation $\tilde{K}\bm{v}=0$ will generally not
have a solution, and instead we will look for a $\bm{v}$ that
minimizes $\|\tilde{K}\bm{v}\|$. This is precisely the smallest
right singular vector (SV) of $\tilde{K}$, which will deviate from
being proportional to the GH because of the noise. As shown in
\cRef{Bairey_2019} using a perturbative analysis, the error in
$\bm{v}$ due to the statistical noise depends on the singular values
of $K$:
\begin{align}
    \label{eq:reconstruction-error}
    \mathbb{E}(\left\| \bm{v}-\bm{v}_{true} \right\|)
      \approx \epsilon \sqrt{\sum_{i>0} \frac{1}{\lambda_{i}}}
      + \orderof{\epsilon^2},
\end{align}
where $\Bv_{true}$ is the exact coefficients vector of the
Hamiltonian, and $\lambda_i$ are the eigenvalues of
$K^T K$. Therefore, a sufficiently large gap
between $\lambda_1$, the second smallest singular value of $K$, and
$\lambda_0=0$ is an necessary condition for the CM method to work
with a reasonable number of measurements.

\section{Hamiltonian Learning Tomography}
\label{sec:ehlt}

\subsection{Description of the algorithm}\label{subsec:description-of-the-algorithm}

In this section, we show how we can combine the Hamiltonian learning
algorithm together with EHT, and introduce the Hamiltonian Learning
Tomography (HLT) method.

To understand the general idea behind it, consider the expansion in
\Eq{eq:H_def}, which defines a vector $\bm{v}$ for any Hamiltonian
$H$ via the $k$-local Pauli basis. This expansion also
defines a Hamiltonian $H_{\bm{v}}$ for every vector $\bm{v}$. Then
by the definition of the $K$ matrix in \Eqs{eq:commutator}{eq:K_nm},
it follows that for every vector $\bm{v}$, the norm
$\norm{K\bm{v}}$ measures the commutativity of $H_{\bm{v}}$ with
$\rho$. When $K$ is slightly perturbed, the Hamiltonians that
correspond to its lowest singular values will be the ones that
nearly commute with $\rho$, and so it is plausible that the actual
Gibbs Hamiltonian is well approximated by a superposition of these
Hamiltonians.

This intuition can be justified using perturbation theory, as shown
in \App{sec:appendix-error-estimation}. Indeed, considering
\Eq{eq:noisy-CM} perturbatively in $\epsilon$, the lowest singular
vector of $K$ will be a equal to the lowest singular vector of
$\tilde{K}$ \emph{plus} a first order correction, which is a
superposition of the other singular vectors of $\tilde{K}$.  The
higher the singular value is, the smaller is the weight of the
corresponding vector in the super position. Specifically, we show in
\App{sec:appendix-error-estimation} that if $\bm{v}_0$ is the zero
singular vector of $K$ that corresponds to the exact Hamiltonian,
and $\tilde{\bm{c}}_l$ is the projection of this vector on the
subspace of the first $l$ singular vectors of the noisy $\tilde{K}$,
then
\begin{equation}
  \label{eq:perturbation-bound}
    \mathbb{E}\big(\norm{\tilde{\bm{c}}_l-\bm{v}_0}^2\big)
    = \epsilon^2\sum_{i\ge l} \frac{1}{\lambda^2_i}
      +  \orderof{\epsilon^3},
\end{equation}
where $\epsilon$ is the statistical noise parameter from
\Eq{eq:noisy-CM} and the averaging $\mathbb{E}(\cdot)$ is defined with
respect to the underlying probability space in which $E$ is defined.

Taking the lowest $l$ right singular vectors of $\tilde{K}$, we can
now follow the EHT method to find the actual coefficients in the
expansion of the actual GH. Formally, we define a $k$-local HLT algorithm with
cutoff $l$ to consist of the following steps:
\begin{enumerate}
    \item For a given (unknown) state $\rho$ prepared on a quantum device,
    measure all $2k$-local Paulis expectation values.

    \item Construct a constraint matrix $\tilde{K}$ as described in
      \Sec{subsec:constraint-matrix} with $k$-local Paulis as
      $\left\{ S_m \right\}$ and $\left( k+1 \right)$-local Paulis
      as $\left\{ A_q \right\}$, and find its $l$ lowest singular
      vectors $\{\tilde{\bm{v}}^{(i)}\}_{i=1}^l$.

    \item Define an ansatz for the GH using the above $l$ vectors
      and use them to introduce the variational parameters
      $\bm{\theta}$,
      \begin{align}
          \label{eq:kl-EHLT}
            H(\bm{\theta}) = \sum_{i=1}^l \theta_i
              \Big(\sum_m \tilde{v}^{(i)}_m S_m \Big),
      \end{align}
     and optimize $\bm{\theta}$ using a measurement-based loss
     function (see \Sec{subsec:optimization}).
\end{enumerate}

For a fixed $k$, the number of $k$-local Paulis scales linearly
in the size of the system. Therefore, the size of the CM and the
number of its singular values are also linear in the number of
qubits, which gives a polynomial description of the underlying Gibbs
state.

Using the fact that the $k$-local Paulis from step (i) sit on
contiguous sites, we can measure many of these observables
simultaneously using the so-called \emph{overlapping local
tomography} technique introduced in \cRef{zubida2021optimized}.
This allows us to use $m$ global measurements to estimate
every $t$-local observable using $m/3^t$ measurement results for
$t=1,\ldots,2k$ --- see \App{sec:overlapping-local-tomography} for
more details. Finally, the same measurement results that were used
in step (i) can be conveniently reused in step (iii) at the
optimization step.

At this point, the reader may wonder why use only $l$ SVs out of all
the possible SVs (e.g.\ $12N-9$ in case of $k=2$ and $1D$ chain).
The reason is that limiting $l$ leads to a much more efficient
optimization and might ease the convergence of gradient-based
methods as a result of lower dimensionality\cc{bubeck2014convex,
bellman1957dynamic,van2009dimensionality}.  Additionally, the
complexity of each gradient computation, e.g., for the optimization
algorithm from \cRef{Branch1999}, scales linearly with the number of
parameters. Moving to larger systems, this is even more pronounced.
For example, an optimization with $10$ qubits using $30$ SVs required a few
hours on a personal computer in our case, while using all the SVs may take up
to several days. This saving in optimization time might become
crucial when one wants to perform tomography many times.

\subsection{Ansatz Optimization}
\label{subsec:optimization}

Once the $l$ singular vectors $\tilde{\bm{v}}^{(i)}$ of $\tilde{K}$
with the smallest singular values are found (see \Sec{sec:ehlt}), we
need to optimize their corresponding weights. Specifically, we want
to find the weights $\bm{\theta}$ such that the state $\rho(\bm{\theta})
= \frac{1}{Z}e^{-H(\bm{\theta})}$, with $H(\bm{\theta})$ defined
in~\eqref{eq:kl-EHLT}, will match the measurements results. To that end,
we use the loss function defined by
\begin{align}
\label{eq:loss-function}
  \chi^2(\bm{\theta})=\sum_B \sum_{\bm{s}}
    \left[ P_B\left( \bm{s} \right)
    - \bra{\bm{s}}U_B^\dagger\rho_A(\bm{\theta})
    U_B \ket{\bm{s}} \right]^2.
\end{align}
The first summation is over the $2k$-local Pauli bases needed for
the construction of the constraint matrix (see
\App{sec:overlapping-local-tomography}). $P_B(\bm{s})$ is the
\emph{empirical} probability to measure the bitstring $\bm{s}$ in
the Pauli basis $B$, which was obtained from the experiment, and
$\bra{\bm{s}}U_B^\dagger\rho_A(\bm{\theta}) U_B \ket{\bm{s}}$ is the
corresponding probability from the ansatz. $U_B$ is a unitary that
rotates the computational basis to the $B$ basis. The second summation is
over all the possible string outcomes of $N$ qubits measurements.

Our loss function is different from the one used in
\cRef{Kokail_2021} in that we do not use random unitaries for $U_B$.
This is advantageous for HLT since it allows one to use the same
measurements for both the calculation of the CM $\tilde{K}$ and the
optimization.  This choice is justified by the fact that $2k$-local
Paulis are tomographically complete, meaning they form a basis of
$2k$-qubits density matrices. Assuming that the underlying state we
learn is described by a $k$-local GH, this means that when the
loss function vanishes, the underlying state has the same $2k$-local
reduced density matrices as the ansatz state $\rho(\bm{\theta})$. By
\cRef{anshu2020sample}, this implies that the two states are
\emph{globally} equal.  Therefore, minimizing the loss function,
using only Paulis basis measurements, is enough.

\subsection{Assessing the HLT performance}
\label{subsec:QA}

A non-trivial question is how to assess the performance of the HLT
method in an actual experiment; how can we tell how close is the HLT-recovered
state to the actual state in the hardware? We suggest two practical
checks that can be performed. First, by doing local
tomography of subsystems with a number of qubits that is larger than
$k$ but feasible for QST, and checking that the resultant reduced
density matrices are close to the corresponding reduced density
matrices of the HLT state. Secondly, one can check
how well the HLT state converges as a function of $l,m$, by
calculating the fidelity of the state that was obtained with the
highest $l,m$ with states of lower $l,m$.  As we demonstrate in
\Sec{subsec:ghz-state-ehlt} (\Fig{fig:ghz-results} bottom), fast
convergence indicates a better performance of HLT.

\section{Results}\label{sec:results}

In this section we present our results, which demonstrate the usage
and performance of HLT on two cases: a chain of $N$ qubits with
either a 1D transverse Ising model Gibbs state, or the reduced
density matrix of a GHZ state. As we show in
\Sec{subsec:ghz-state-ehlt}, the latter is a Gibbs state of the
classical Ising model at $T=0$. In both cases, $k=2$. We present
both numerical simulations results and results of IBM Quantum
devices accessed via the cloud\cc{IBMQ}, implemented using the
Qiskit framework\cc{qiskit}.
Simulations and data processing were done on a personal computer
with $32$Gb RAM memory.
To assess the quality of our method, we
use the quantum fidelity, defined by $F(\rho,\sigma) \EqDef
(\Tr\sqrt{\sqrt{\rho} \sigma \sqrt{\rho}})^2$\cc{Uhlmann1976,
Jozsa1994}.

When running HLT with $m$ measurements, those were distributed equally
between the $3^{2k}$ basis-measurements that were needed to measure
the expectation values of $2k$-local Paulis (see \Sec{sec:ehlt} and
\App{sec:overlapping-local-tomography}). On the other hand, when $m$
measurements were used for ordinary QST on $N$ qubits, the measurements were
distributed equally between the $3^N$ measurement bases. More
details on the QST procedure can be found in \App{sec:appendix-QST}.
When all SVs were used for HLT, $l$ is denoted by $l_{max}$.

In our analysis, the optimization was done using
SciPy's\cc{2020SciPy-NMeth} nonlinear least-squares solver (with
2-point Jacobian computation method) which uses the algorithm from
\cRef{Branch1999}. We also implemented gradient decent optimization
method (using the same loss function) with
PyTorch\cc{NEURIPS2019_9015} and were able to reproduce the results
from the first optimization method, validating the optimization
results.

\subsection{Transverse Ising Gibbs States}
\label{subsec:gibbs-state-ehlt}

\begin{figure}
\label{fig:theory-results-1} \centering
  \includegraphics[width=0.48\textwidth]
    {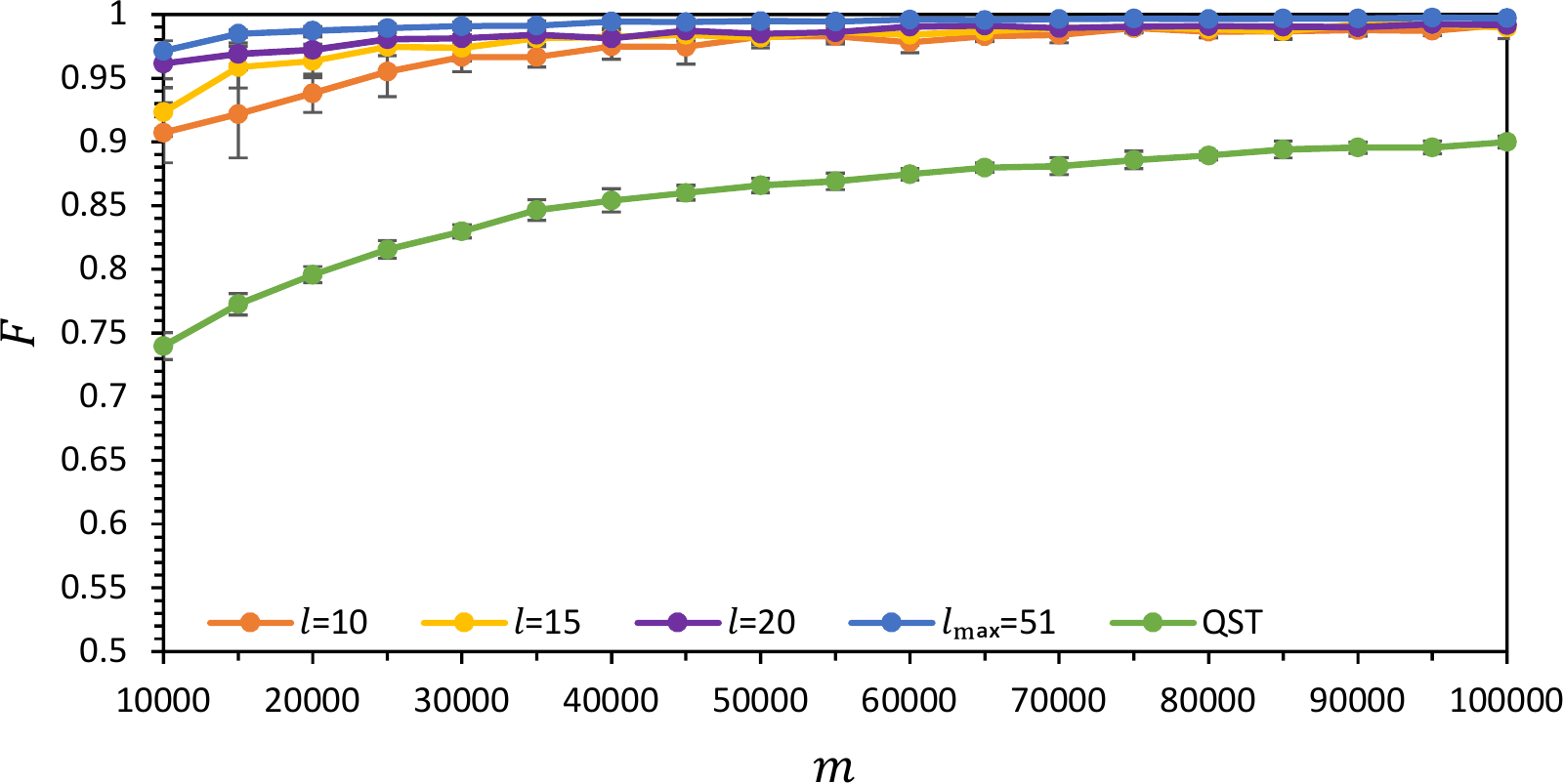}
  \caption{Results of a numerical simulation checking the HLT
    performance on an \emph{exact} Gibbs state of the Transverse
    Ising model on 5 qubits with $k_B T=1$. The plot shows the
    fidelity between the exact Gibbs state  and the HLT-reconstructed
    state as a function of the total number of measurements ($m$)
    for different numbers of singular values ($l$). The green line
    shows the fidelity of the exact state with a QST-reconstructed
    state using $m$ measurements. The plot shows the clear
    advantage of the HLT approach in this case. Results are averaged
    over $10$ simulation runs, error bars calculated using standard deviation.}
\end{figure}

\begin{figure*}
    \label{fig:theory-results-2}
    \centering
    \hfill
    \begin{minipage}{0.015\textwidth}
        \vspace{-110pt}
        \textbf{a)}
    \end{minipage}
    \begin{minipage}{0.47\textwidth}
        \includegraphics[width=1\textwidth]{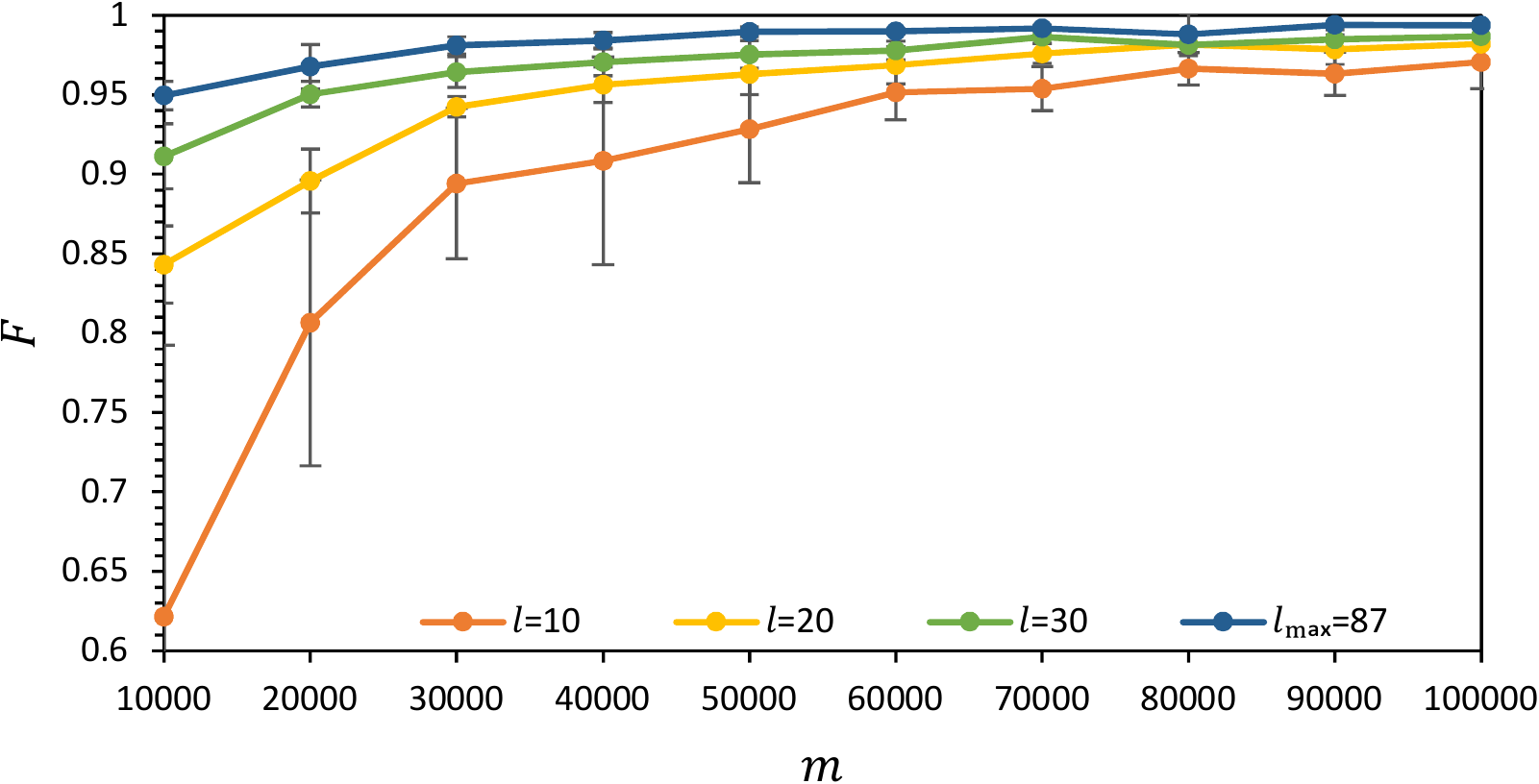}
    \end{minipage}
    \hfill
    \begin{minipage}{0.015\textwidth}
        \vspace{-110pt}
        \textbf{b)}
    \end{minipage}
    \begin{minipage}{0.47\textwidth}
        \includegraphics[width=1\textwidth]{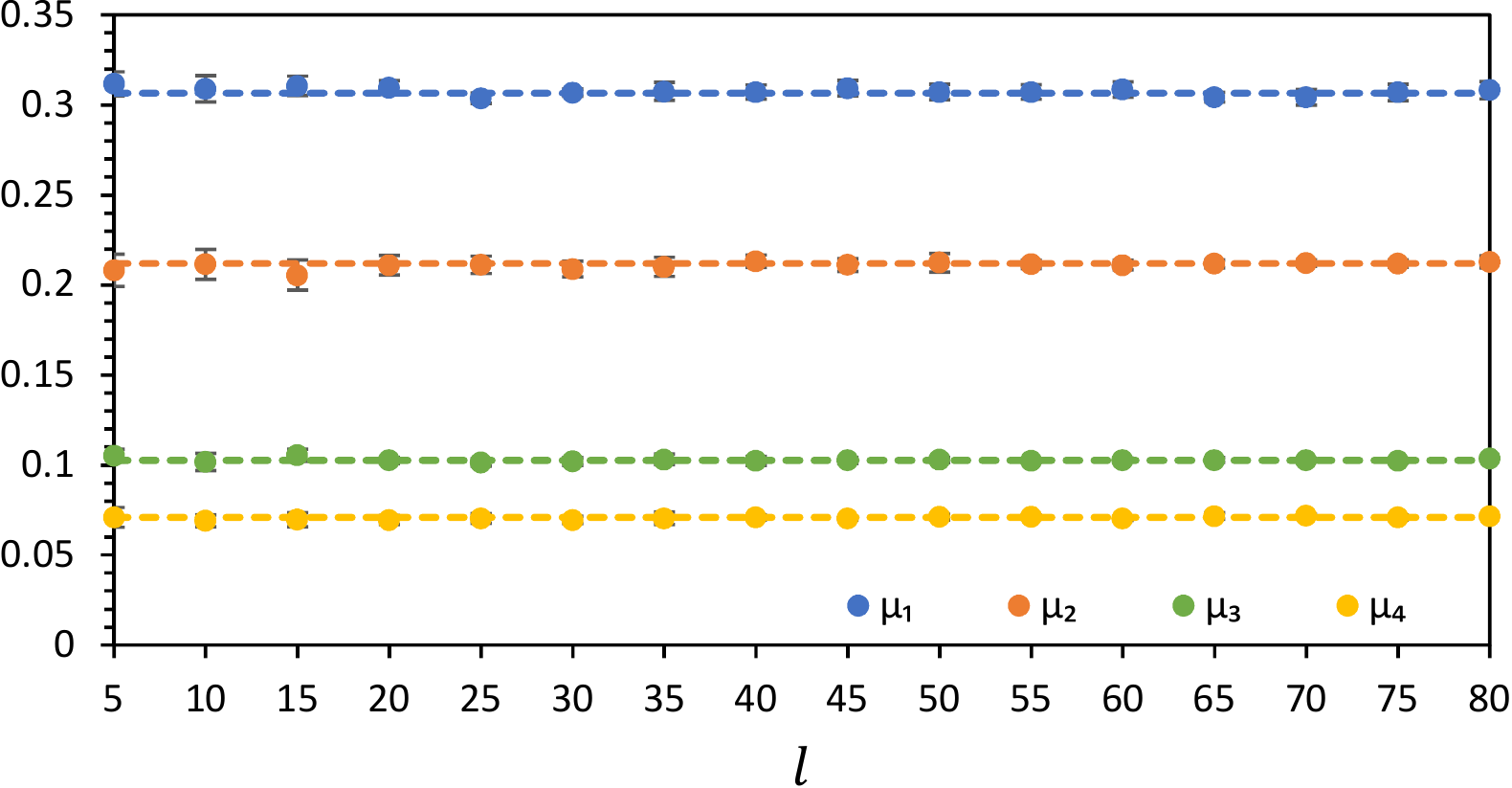}
    \end{minipage}
    \caption{Numerical simulation results of HLT of an exact Gibbs state of $8$
      qubits for the transverse-field Ising model.
      Results are averaged over $10$ simulation runs, error bars calculated
      using standard deviation. \textbf{a)} Fidelity between
      the exact Gibbs state and HLT-reconstructed state as a function of the
      number of measurements ($m$), for different numbers of singular values
      ($l$). \textbf{b)} The performance of HLT on
      non-local, non-linear quantities. The plot shows the
      four highest eigenvalues $\mu_i$ of the HLT-reconstructed
      state with $m=10^5$ as a function of the number of singular values ($l$).
      Dashed lines denote the exact eigenvalues of the Gibbs state.
      The plot shows excellent reconstruction around $l \gtrsim 20$.}
\end{figure*}

The 1D Transverse Ising Hamiltonian is given by:
\begin{align}
\label{eq: transverse ising}
  H_{TI}=\sum_{\left< i,j \right>}X_i X_j + \sum_i Z_i.
\end{align}
In this section, we study the case $k_B T=1$, which leads to the
Gibbs state:
\begin{align}
\label{eq:ising-dm}
  \rho_{TI} &= \frac{1}{Z}e^{-H_{TI}}, &
  Z &= \Tr\left( e^{-H_{TI}} \right).
\end{align}

\subsubsection{Exact Gibbs state simulations}

We first studied the performance of HLT in numerical simulations of
an exact Gibbs state of the transverse-field Ising model. We started with
a $5$ qubits state, and simulated the HLT protocol on it using
$l=10,15,20$ and $l_{max}=51$ SVs and a total number of measurements
that was varied from $m=5,000$ to $m=10^5$. For comparison, we also
performed a traditional QST using the same number of measurements
(but with different measurement bases --- see \App{sec:appendix-QST}
). Our results are shown in
\Fig{fig:theory-results-1}, where the fidelities of the different
protocols with the exact Gibbs state are plotted.
The plot clearly demonstrates the advantage of HLT
over traditional QST for the case of an ideal local Gibbs state. For
HLT, fidelities of over $0.9$ and $0.97$ were obtained using only
$10^4$ and $5\times 10^4$ measurements, respectively, for $l\ge 15$.
On the other hand, QST fidelities were lower than $0.8$ for $10^4$
measurements and reached maximal fidelity of about $0.91$ with
$10^5$ measurements.
Figure~\ref{fig:theory-results-1} also clearly demonstrate the
improvement of the fit as we add more SVs; it shows that in this
case, using $l=20$ SVs gives comparable results to $l_{max}=51$.

We have also tested the performance of HLT, by repeating the
numerical experiment on a larger number of $8$ qubits. Here, we did
not perform a QST, as it is not realistic on present day quantum
hardware. Figure~\ref{fig:theory-results-2}a shows the fidelity of
the HLT with the true Gibbs state as a function of number of
measurements $m$ for $l=10,20,30,40$ and $l_{max}=87$. Fidelity of
over $0.9$ was achieved with merely $l = 30$ and $m=
2\times 10^4$. The $l=10,20$ results showed an instability that can be
attributed to the instability of the subspace spanned by the lowest
SVs of $\tilde{K}$ for the rather low number of measurements that we
used. This instability is reduced by either by increasing $m$, or by
increasing $l$ and thereby increasing the probability of a large
overlap between the subspace of the lowest SVs and the true Gibbs
Hamiltonian.

In \Fig{fig:theory-results-2}b we have tested the ability of HLT to
infer \emph{global}, non-linear properties of the state. The plot
shows the highest eigenvalues the HLT state, which was reconstructed
using $10^5$ measurements using different numbers of SVs $l$. The
straight dashed lines show the exact values of these eigenvalues.
 HLT manages to reconstruct these non-local and
highly non-linear quantities with an accuracy of $0.01$ for
$l\geq 20$.

\subsubsection{Noisy Gibbs state simulations}
\label{subsec:Variational Gibbs State}

To test the performance of HLT on more realistic scenarios, we
numerically simulated the variational quantum algorithm VarQITE of
\cRef{Yuan2019} for creating a Gibbs state.  Starting with a
maximally entangled state between two sets $A$ and
$B$ of qubits, the algorithm approximates an imaginary time evolution of
$e^{-H/2}$ on one set, with $H=H_A\otimes\Id_B$, so that the reduced
density matrix on $A$ becomes the Gibbs state $\frac{1}{Z}e^{-H_A}$.
The circuit that approximates $e^{-H/2}$ consists of C-NOT gates and
parametrized single qubit rotations about the $Y$ axis that approximate
the Torreterized imaginary time evolution. The weights in this
circuit are determined variationally using a gradient-based method
and discrete steps. Technical details of the algorithm can be found
in \App{sec:appendix-VarQITE}.

To benchmark HLT, we simulated a noisy implementation of the VarQITE
algorithm on 5 qubits using the Qiskit simulator with a noise model, where the
noise parameters (such as $T_1,T_2$ times and 1,2-qubit gate errors)
were determined from actual quantum hardware.
Comparing to the ideal, noiseless simulations of the previous
subsection, the above algorithm introduces several sources of errors
that can challenge the HLT. First, it is easy to see that even an
ideal, noiseless implementation of the algorithm will still not
create an exact Gibbs state due to the unavoidable Trotter
errors of a low-depth circuit. In addition, the circuit itself is
found using a variational method, which might introduce optimization
errors. Finally, on top of that, there are the errors caused by the
noisy simulation. All of these effects will cause the Gibbs
Hamiltonian to deviate from the ideal transverse Ising model
Hamiltonian, and will introduce some less local terms to it (such as
3-local or 4-local terms). In our case, we wanted to check how
HLT with $k=2$ can cope with this type of state.

\begin{figure}
\label{fig:noisy-simulation-results} \centering
  \includegraphics[width=0.48\textwidth]
      {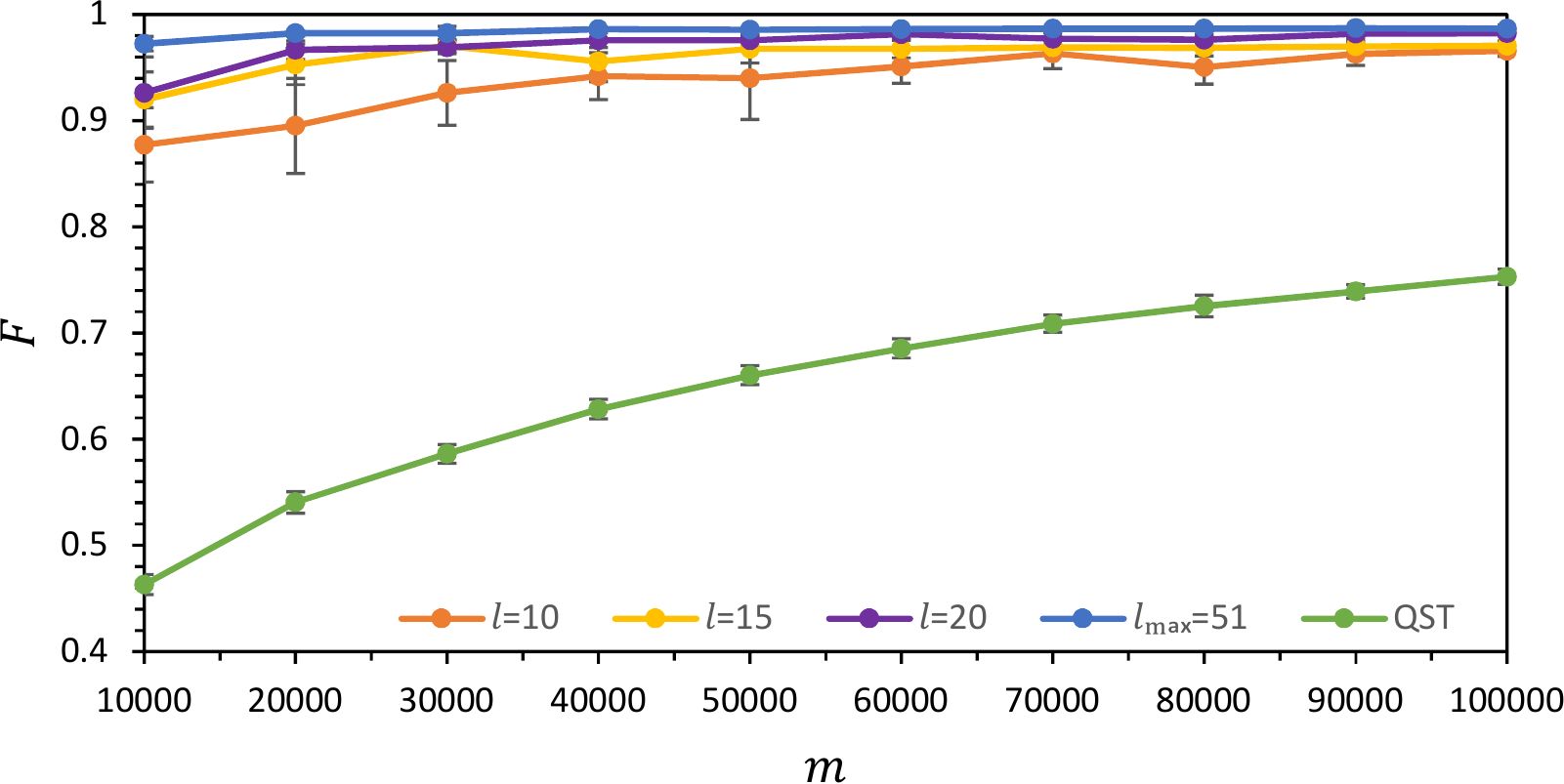}
  \caption{Performance of HLT on a 5-qubits Gibbs state created by a
  simulation of the VarQITE algorithm using Qiskit
  backend with a noise model. The plot shows the Fidelity between the exact
  simulated state and the HLT/QST-reconstructed state as a function of the total
  number of measurements ($m$). The number of SVs ($l$) used by HLT is
  shown in the legend. Measurements for QST were distributed equally over
  all $3^5$ different Pauli bases. Results are averaged over $10$ simulation
  runs, with error bars indicating one standard deviation.}
\end{figure}

In \Fig{fig:noisy-simulation-results} we present the
performance of HLT 5-qubits noisy numerical simulations. As in
\Fig{fig:theory-results-1}, we plot the fidelity of the HLT states
with the actual state in the simulation. The plot shows HLT
performance for different numbers of SVs and measurements $m$, and
compares it with a traditional QST with the same number of measurements.
HLT performed well even though the GH of the state was not
entirely 2-local, because of the noise in the simulation.
To be precise, the 3-local or more Pauli coefficients were $4.68\%$
of the GH norm (see \App{sec:appendix-hardware-state-preparation}).

\subsubsection{Quantum hardware results}

In this subsection we present benchmarks of the HLT on actual
quantum hardware. As in the previous section, we used the VarQITE
algorithm to generate an approximate transverse Ising model Gibbs
state on $N$ qubits by approximate imaginary time evolution to a
maximally entangled state of $2N$ qubits.

Figure~\ref{fig:hw-results-gibbs-cutoff} presents the results of a
$5$ qubits HLT on the \emph{ibmq\_mumbai} $27$-qubits backend~\footnote{The
experiment was conducted on the device \emph{ibmq\_mumbai} on 23/11/2021.
Backends are listed in \url{https://quantum-computing.ibm.com/}}.
Unlike the numerical simulation case, here we did not know the exact
quantum state. Therefore, in order to assess the quality of HLT
state, we compared it to a high-precision QST state that was
obtained using a much larger set of measurements: we used $2\times
10^4$ measurements on each of the $3^5$ QST bases, which amounts to
almost $5$ millions measurements. In \App{sec:appendix-QST} and
Table~\ref{tab:qst-fidelity} we use numerical noisy simulations to
estimate that for this number of qubits and measurements, the QST
fidelity from the real quantum state is about $0.985$.

Figure~\hyperref[fig:hw-results-gibbs-cutoff]{4a} shows the fidelity of the
QST state with the state reconstructed by HLT with merely
$10^4$---$10^5$ measurements and $l=10,30$. Fidelity exceeded $0.98$
with $2\times 10^4$ measurements or more. We also used QST with
diluted number of measurements, which achieved a fidelity of less
than $0.7$ with the high-precision QST, using a total of $10^5$
measurements (Fig.~\hyperref[fig:hw-results-gibbs-cutoff]{4a}, yellow points). In
order to further verify the above results, we extracted the four
highest eigenvalues of HLT and QST reconstructed density matrices
(Fig.~\hyperref[fig:hw-results-gibbs-cutoff]{4b}). Agreement of $0.01$ with the
high-precision QST result was obtained with merely $25$ SVs,
demonstrating the ability of HLT to predict global, non-linear
properties of the state on quantum hardware.

\begin{figure*}
    \label{fig:hw-results-gibbs-cutoff}
    \centering
    \hfill
    \begin{minipage}{0.015\textwidth}
      \vspace{-100pt}
        \textbf{a)}
    \end{minipage}
    \begin{minipage}{0.47\textwidth}
        \includegraphics[width=1\textwidth]{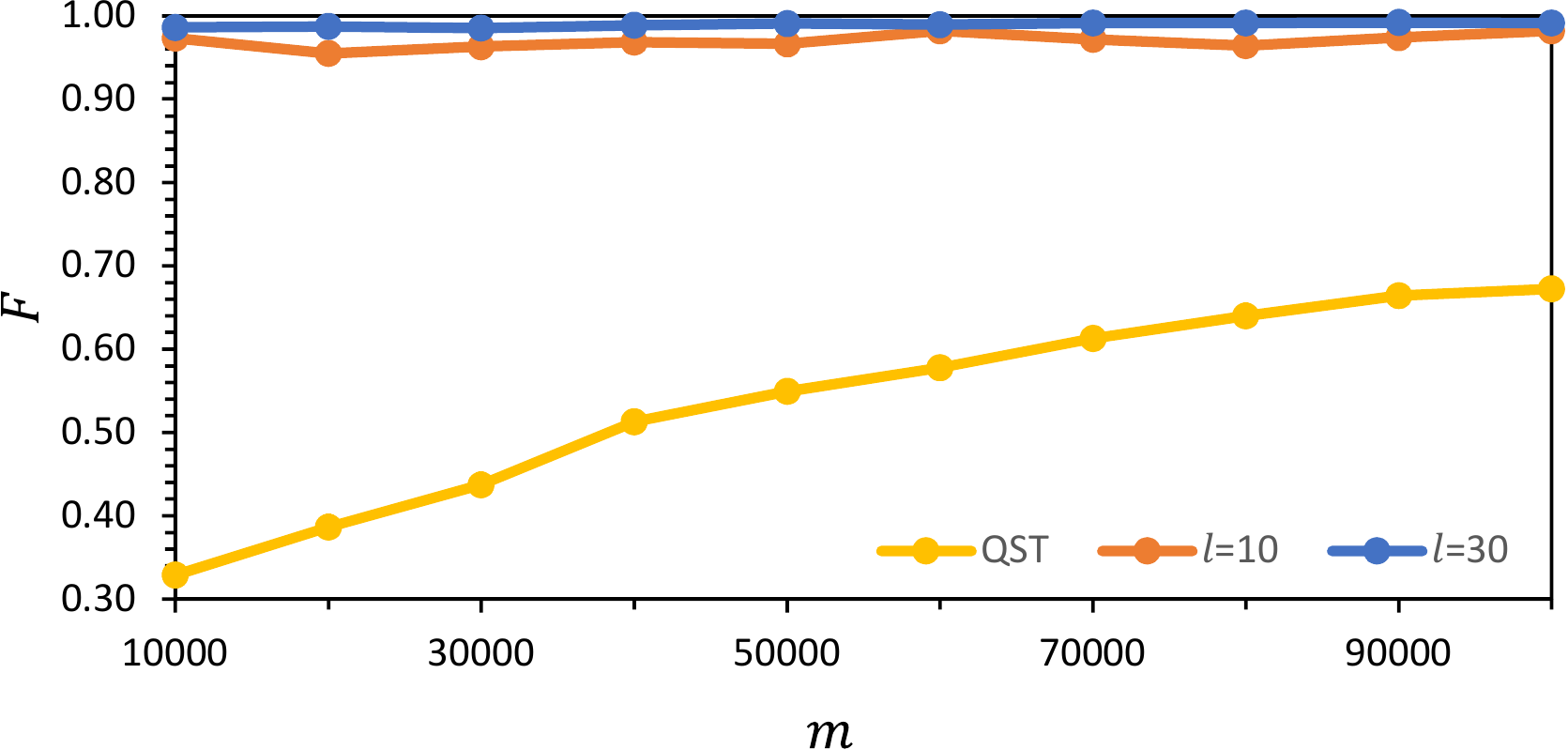}
    \end{minipage}
    \hfill
    \begin{minipage}{0.015\textwidth}
        \vspace{-100pt}
        \textbf{b)}
    \end{minipage}
    \begin{minipage}{0.47\textwidth}
        \includegraphics[width=1\textwidth]{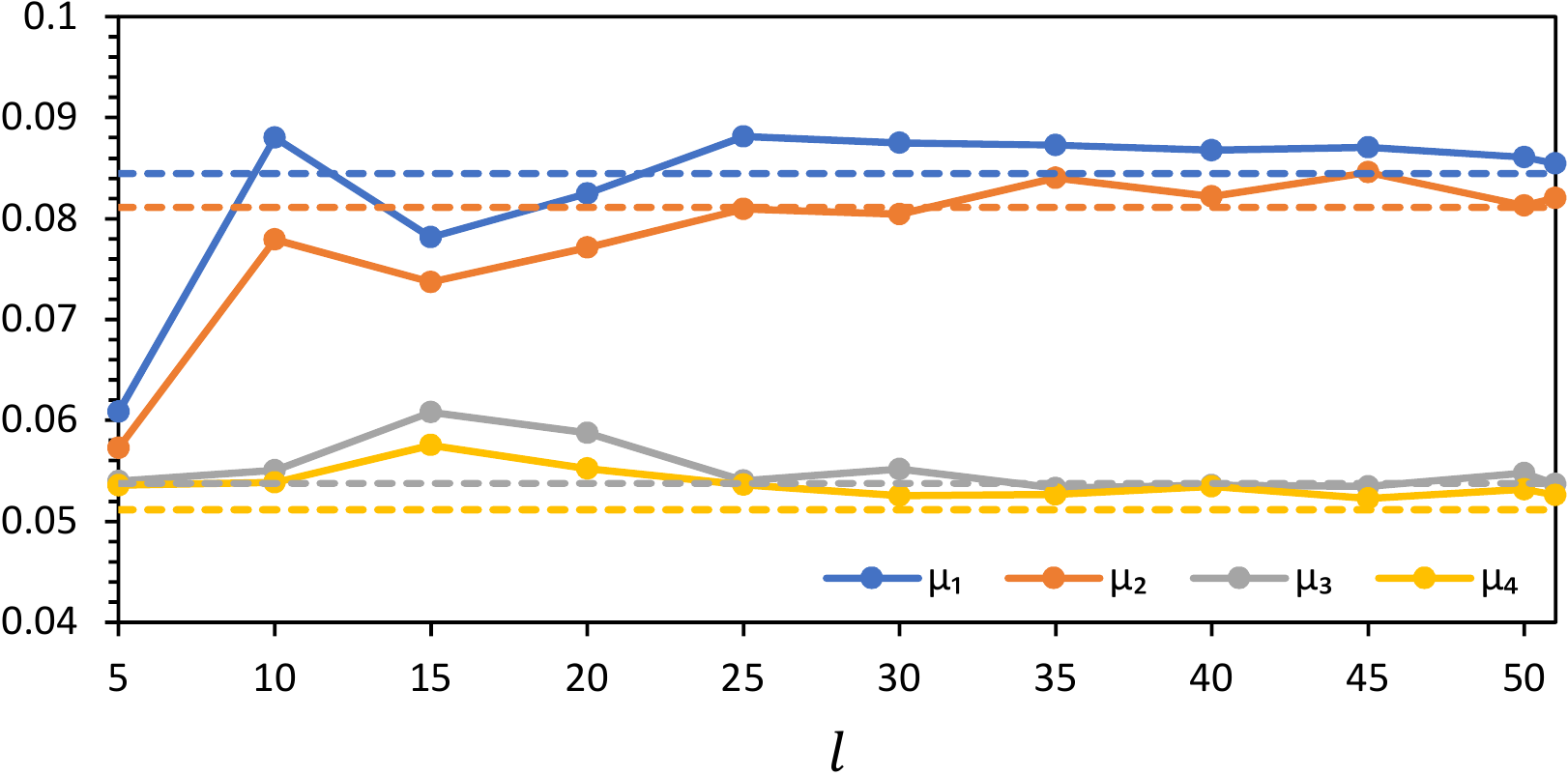}
    \end{minipage}
    \caption{Quantum hardware results of HLT and QST using the
    \emph{ibmq\_mumbai} 27-qubits backend over $5$ qubits in a Gibbs state of
    the transverse-field Ising model.
    \textbf{(a)} The fidelity of HLT and QST as a function of
    the number of measurements $m$. The fidelity was calculated with
    respect to a high-precision QST, which used $2\times 10^4$ measurements
    in each of the $3^5$ QST bases (with estimated fidelity of above
    $0.985$). Subsets of the same measurements were used for QST and HLT in
    order to eliminate bias due to temporal fluctuations of device
    parameters. The plot shows a clear advantage of the HLT method,
    already with $l=10$ SVs. \textbf{(b)} The four
    highest eigenvalues, $\mu_i$, of HLT-reconstructed state with
    $10^5$ measurements as a function of $l$. Dashed
    lines show the four highest eigenvalues of the high-precision QST result.}
\end{figure*}

Figure~\ref{fig:partial-verification-results} shows the results of
HLT on larger systems of $4-10$ qubits~\footnote{The
experiment was conducted on the device \emph{ibmq\_toronto} on 06/11/2021.}.
Just as in the $5$ qubits case, we used the VarQITE algorithm to approximate
the transverse model Ising Gibbs state on these systems. For $N>5$, QST of the
entire system is not feasible because of the large number of
measurements needed. Instead, to evaluate the HLT performance we
used QST to reconstruct the states of all $3$-qubit contiguous
subsystems (i.e $\left\{ 1,2,3 \right\},\ldots,\left\{ N-2,N-1,N
\right\}$) and calculated the average fidelity between them and
HLT-reconstructed states. The local QST was done using $8,192$
measurements in each of the QST $3^3$ bases, and as we argue in
\App{sec:appendix-QST}, the expected QST fidelity in this case is
above $0.995$. As shown in \Fig{fig:partial-verification-results},
average local fidelity results between the HLT and the local QST was
above $0.99$ with $l=30$ and a modest number of $m=10^5$
measurements for any system size. For $N=10$ qubits, the low number
of SVs used ($30$ out of $l_{max}=111$) was especially important for
the optimization step, which would have taken days had we used the
full $l_{max}$ SVs.

We end this subsection by noting that the actual Gibbs states
prepared on the quantum hardware differed substantially from the
ideal Transverse Ising model Gibbs state. This can be attributed to
imperfections in the quantum hardware and the fact that the
optimization of the variational algorithm was done on a classical
simulation and not on the quantum hardware itself. As described in
detail in \App{sec:appendix-hardware-state-preparation}, the QST and
HLT results show that the underlying Gibbs Hamiltonians were
classical to some extent: their most dominant terms in the Pauli
expansion were in the $Z$ basis. The HLT method is, of course, blind
to this bias. The fact that it performed well on local Gibbs states
that were different from the intended ones, demonstrates its
capabilities to reconstruct unknown states with local GHs.

\begin{figure}
    \centering
    \includegraphics[width=0.48\textwidth]{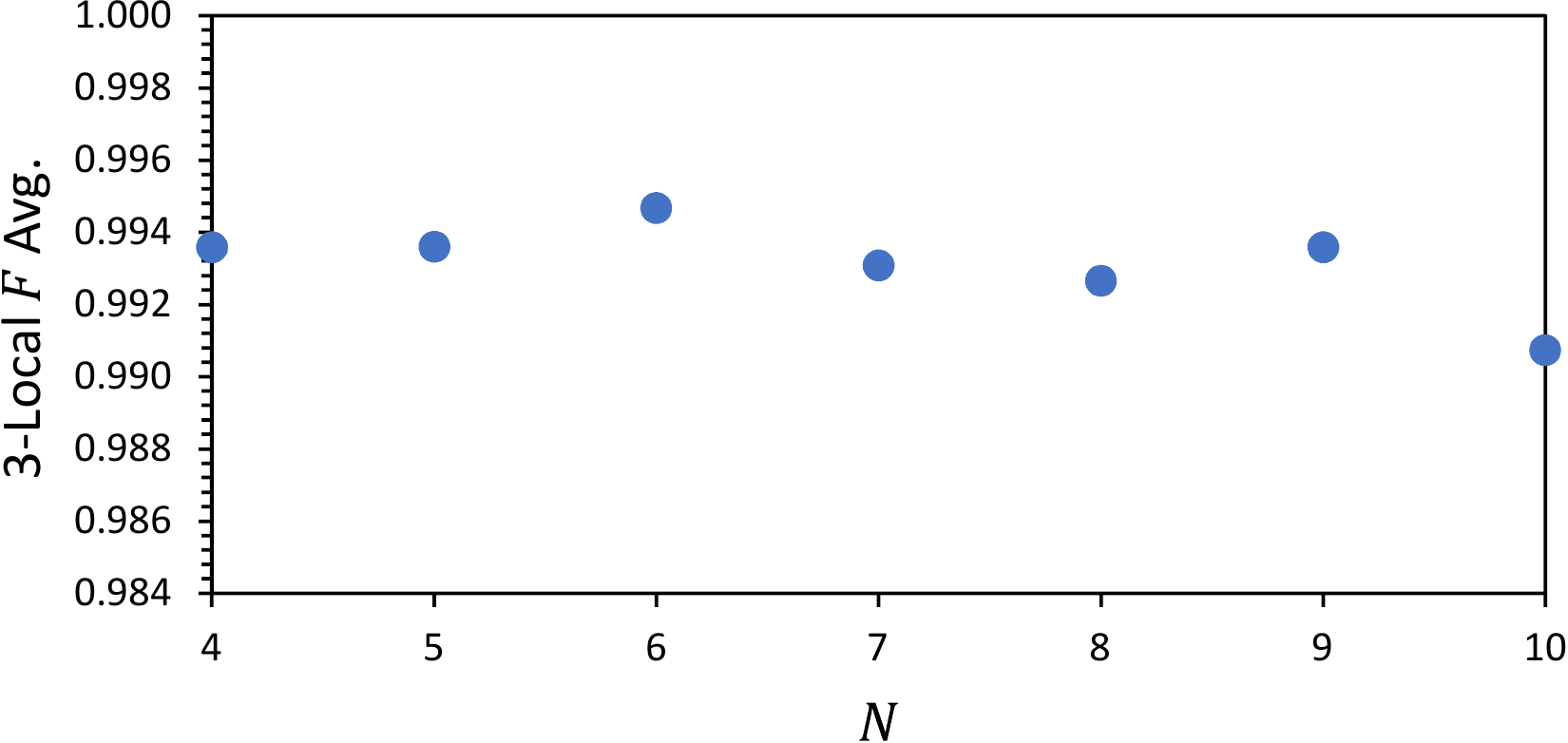}
    \caption{Quantum
    hardware results of HLT using the \emph{ibmq\_toronto} 27-qubits
    backend, for transverse-field Ising
    Gibbs states with $4$---$10$ qubits.
    The plot shows the average fidelity, over all
    $3$-qubits contiguous subsystems, between QST and HLT-reconstructed
      states as a function of the number of qubits $N$.
    QST was run separately on each subsystem, whereas using HLT with
    $10^5$ measurements and $30$ SVs the full state was reconstructed,
    from which the reduced density matrices of the subsystems were extracted.
    QST for each subsystem used $8,192$ measurements in each
    of the $3^3$ local bases (with an estimated fidelity above $0.995$).}
    \label{fig:partial-verification-results}
\end{figure}

\subsection{GHZ state}
\label{subsec:ghz-state-ehlt}

A \emph{GHZ} state\cc{Greenberger1989}, also called \emph{cat
state}, on $N$ qubits is defined as:
\begin{align*}
  \frac{\left| 0\right>^{\otimes N}
    + \left|1\right>^{\otimes N}}{\sqrt{2}}.
\end{align*}
While the GHZ is not by itself a local Gibbs state, tracing out one
or more qubits gives the reduced density matrix
\begin{align}
\label{eq:ghz-dm} \rho_{GHZ} = \frac{1}{2}\left(
    \left|0\ldots0\rangle \langle 0\ldots0\right| +
    \left|1\ldots1\rangle \langle 1\ldots1\right| \right) .
\end{align}
This is the $T=0$ Gibbs state of the classical Ising Hamiltonian,
and so by using
\begin{align}
    \label{eq:h-ghz}
    H_{GHZ} = - a\sum_{i=1}^{M-1}Z_i Z_{i+1}\quad,\quad a\gg 1.
\end{align}
we get
\begin{align*}
    e^{-H_{GHZ}}=e^{aM}\bigg[& \left|0\ldots0\rangle \langle 0\ldots0\right| +
    \left|1\ldots1\rangle \langle 1\ldots1\right| +\\
    & +\sum_{j=2}^{2^M-1} e^{-r_j a}  \left|j\rangle\langle j\right| \bigg]
    \quad,\quad \forall j:\,\,r_j\ge1 ,
\end{align*}
which after normalization is a good approximation of $\rho_{GHZ}$
for large enough $a$.

We finally note that unlike the Gibbs state of the previous section,
a GHZ state is an extremely ``fragile'' state, and even a small
amount of noise can cause its reduced density matrix to drift away
from a local Gibbs state. We see this effect in the following
subsection, where we present quantum hardware HLT results of this
state.

\subsubsection*{Quantum hardware results}
\begin{figure}
\label{fig:ghz-results} \centering
  \includegraphics[width=0.49\textwidth]{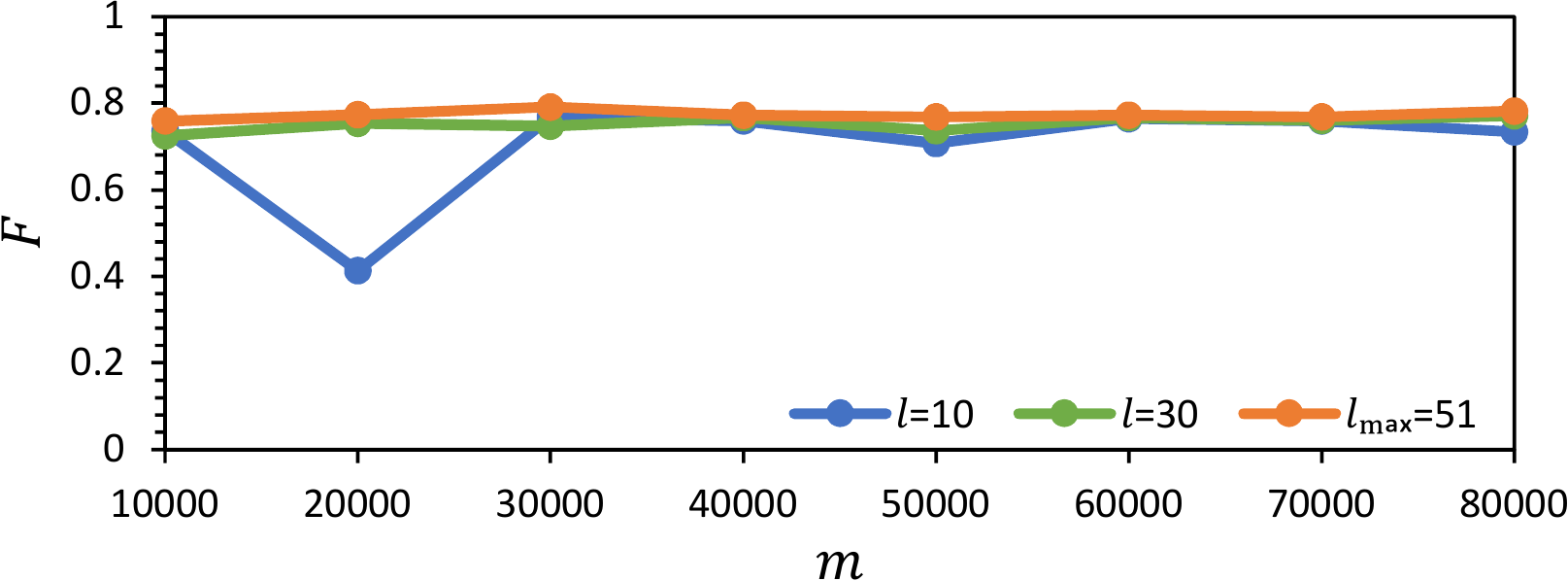}\\
  \includegraphics[width=0.49\textwidth]{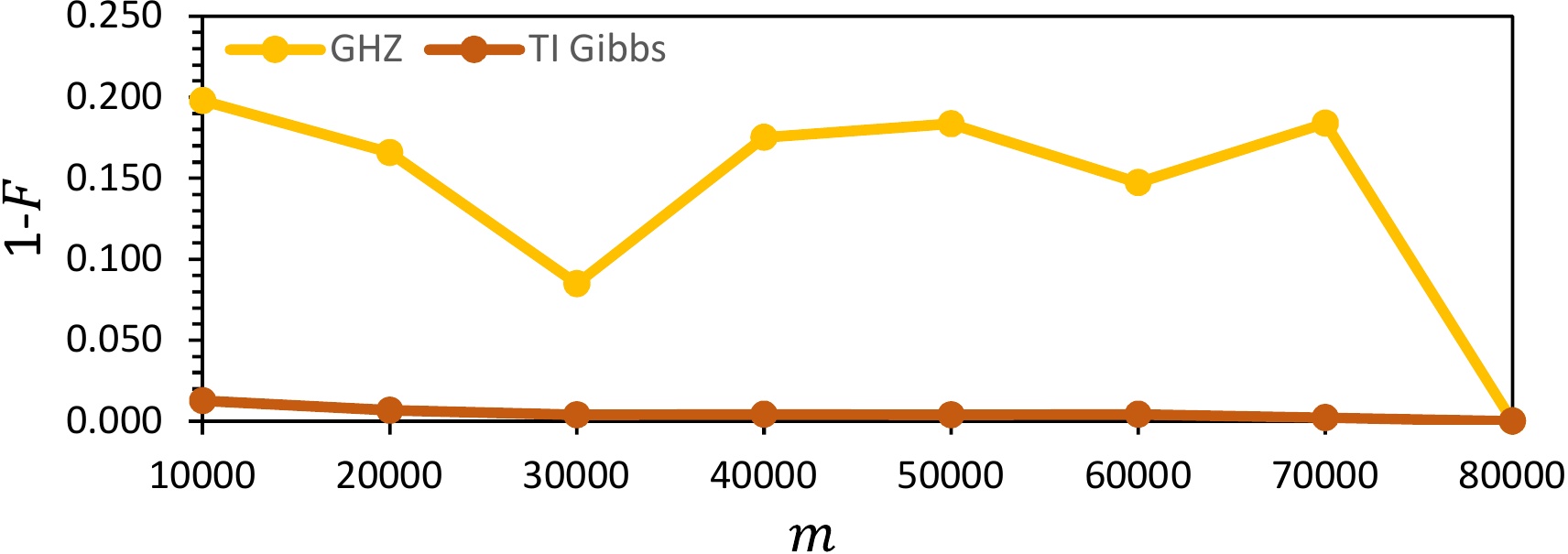}

  \caption[Caption for LOF]{Quantum hardware results for the GHZ case. The
    plots show the HLT performance on the reduced density matrix of $5$
    out of $6$ qubits in a GHZ state prepared on the $ibmq\_mumbai$
    27-qubits backend. Top: fidelity between the
    QST and the HLT reconstructed states. QST used $8,192$ measurements
    in each of the $3^5$ Pauli bases. HLT used between $10^4$ to
    $8\times 10^4$ measurements and different numbers of SVs as described in
    the legend. Bottom: convergence of infidelity for the
    GHZ states (yellow), with the transverse Ising Gibbs states shown as
    reference (brown). Points show infidelity between the HLT-reconstructed
    state with $8\times 10^4$ measurements and the HLT-reconstructed states
    using $m$ measurements, both with $l_{max}=51$ SVs. }
\end{figure}

Using the \emph{ibmq\_mumbai} 27-qubits backend we tested HLT on GHZ
states on quantum hardware for a $5$-qubits reduced density matrix out
of $6$ qubits~\footnote{The experiment was conducted on the device
\emph{ibmq\_mumbai} on 26/08/2021.}. The GHZ state itself was created
using a standard low-depth circuit that contained a Hadamard gate and 5
CNOT gates. Figure~\ref{fig:ghz-results} (top) shows the fidelity between
HLT-reconstructed states and a high-precision QST state as a
function of the number of HLT measurements. The HLT states were
calculated with $l=10,20,30$ SVs and $m$ varying from $10^4$ to $8\times 10^4$
measurements. The QST state was calculated using $8,192$
measurements on each of the $3^5$ different QST bases (a total of
almost two million measurements). The same measurement results were
used for both QST and HLT (with proper dilution) in order to
eliminate bias due to temporal fluctuations of device parameters.
Following Table~\ref{tab:qst-fidelity} in \App{sec:appendix-QST}, we
estimate the QST fidelity in this case to be more than $0.99$.

As clearly shown in the top panel of \Fig{fig:ghz-results}, the
fidelity between the HLT state and the high-precision QST state is
below $0.8$. This is in sharp contrast with the Gibbs state of the
transverse Ising model Hamiltonian shown in
Fig.~\hyperref[fig:hw-results-gibbs-cutoff]{4a}, where the HLT fidelity with
$l=10,30$ quickly went above $0.95$.

As mentioned above, the relative failure of $k=2$ HLT can be
attributed to the hardware noise, which, together with the fragility
of the GHZ state, seems to introduce more $3$-local terms (or
higher) to the Gibbs Hamiltonian. A possible solution  could be
increasing the locality of the Gibbs Hamiltonian in the method from
$k=2$ to $k=3$ or $k=4$, which is however beyond the scope of the current work.

At this point, it is important to understand if it is possible to identify the
failure of the $k=2$ HLT \emph{without} relying on the high-precision QST
results. Following the discussion in \Sec{subsec:QA}, we compared the average
local fidelity of HLT state with $l=l_{max}$ and $m=8\times 10^4$ with the local
QST state. Also here, the fidelity never exceeded $0.91$, indicating
a low global fidelity. Additionally, in \Fig{fig:ghz-results}
(bottom) we studied the \emph{convergence} of HLT on GHZ states as a
function of the number of measurements. The plot shows the
infidelity of the HLT state with $m$ measurements and $l=l_{max}$
with the HLT state that was obtained with the maximal $m=8\times
10^4$ measurements and the same $l$. The figure (yellow line) clearly shows
strong fluctuations in terms of $m$, which is yet another strong indication
for the failure of the $k=2$ HLT. For comparison, we have plotted
the same calculation for the previous case of the VarQITE generated
Gibbs state (brown line). Unlike the Gibbs state, here a clear
convergence is observed as soon as $m>2\times 10^4$.

\section{Discussion and Outlook}\label{sec:discussion-and-outlook}

We suggest \emph{Hamiltonian Learning Tomography}, a general
tomography method for Gibbs state with local Hamiltonian,
which uses moderate computational and quantum resources.
It exploits Hamiltonian learning based on a constraint matrix for observables,
to get a parametrized ansatz for the Gibbs Hamiltonian composed of the
constraint matrix's singular values. The ansatz is then optimized
with respect to local measurement results, similar to EHT.

The use of the Hamiltonian learning
algorithm to generate an ansatz removes the need for an
a-priori knowledge of the Hamiltonian that parametrizes
the state, as relied upon in EHT.
We have shown in multiple examples that the ansatz requires a very
low number of measurements, in comparison with common QST methods,
for reaching a high tomography precision.
Furthermore, using an optimization method
circumvents the need for a gap in the constraint matrix and furthermore, the
normalization of the state, which is enforced by the
optimization method, provides the multiplicative factor of the
Hamiltonian.

We implemented HLT and demonstrated it on states with a local GH,
both in simulations and on IBM Quantum superconducting devices.
At first, we verified HLT for $5$-qubit transverse-field Ising Gibbs state
using high-precision QST.
Then, the density matrix of states with $4$ to $10$ qubits
on the quantum device were reconstructed with high fidelities, using a very low
number of measurements in comparison with common QST methods.
We have also provided two indicators for the performance of HLT that
can be assessed from the experiment data. The first indicator is obtained by
performing local QST for feasible sizes of qubit subsystems, and the second is
the convergence of the HLT states. Both indicate that HLT did not perform
well on the GHZ states created with the quantum device, due to a breaking
of the locality of the Gibbs Hamiltonian.

Considering possible extensions of our results, we here focused on 1D Gibbs
Hamiltonians with $k=2$ locality, but HLT can be
easily implemented on systems with higher $k$ or spatial dimensions.
We note that for the optimizations we ran, we have to be able to hold
the entire $N$-qubits density matrix in the computer memory and
manipulate it, making it suitable for intermediate sized systems of
$2$---$14$ qubits. This is the main bottleneck of the method
for scaling to a large number of qubits. An interesting direction
for future research would be to lift this restriction, for example, either by
some form of efficient exponentiation (i.e.~low-order polynomial
expansions for high-temperatures), or by using tensor
networks\cc{torlai2020quantum}, or both\cc{vanhecke2021simulating}. The
efficiency of the presented tomography method in terms of the number of
measurements, and the fact that it is easy to identify based on the experiment
results states that fall outside of its scope, make it a useful tool for the
verification of quantum states that are hard to characterize.

\section{Acknowledgments}\label{sec:acknowledgments}

We thank Raz Firanko and Netanel Lindner for enlightening
discussions. IA acknowledges the support of the Israel Science
Foundation (ISF) under the Research Grants in Quantum Technologies
and Science No.~2074/19, and the joint NRF-ISF Research Grant
No.~3528/20.

\bibliography{main}
\bibliographystyle{apsrev4-2}


\appendix


\section{GH Reconstruction Error Estimation}
\label{sec:appendix-error-estimation}

In this appendix we derive the bound on the reconstruction error of
the Gibbs Hamiltonian given in~\eqref{eq:perturbation-bound}.
This bound reveals how the reconstruction error
decays with $l$ and therefore justifies using the \emph{lowest}
$l$ singular vectors of the CM in the
ansatz.

\begin{theorem}
\label{thm:perturb} Let $K$ be the exact constraint matrix of some
  Gibbs state $\rho$, and let $\tilde{K}=K+\epsilon E$ be the
  empirical constraint matrix, where $E$ is a random matrix with
  elements that are independent and identically distributed random
  normal variables and $\epsilon \ll 1$. Let $\bm{v}_0$ be the
  singular value of $K$ with the lowest singular value, and let
  $\tilde{\bm{c}}_l$ be the projection of $\bm{v}_0$ on the subspace
  of the first $l$ singular vectors of $\tilde{K}$. Finally, let
  $\lambda_i$ be the $i$'th singular value of $K$, sorted in
  ascending order.  Then
  \begin{equation}
  \label{eq:perturbation-bound-theorem}
    \mathbb{E}\big(\norm{\tilde{\bm{c}}_l-\bm{v}_0}^2\big)
    = \epsilon^2\sum_{i\ge l} \frac{1}{\lambda^2_i}
      +  \orderof{\epsilon^3},
  \end{equation}
  where the averaging $\mathbb{E}(\cdot)$ is defined with respect to
  the underlying probability space in which $E$ is defined.
\end{theorem}

\begin{proof}
  To simplify the
  presentation, we move to Dirac notation and denote the right
  singular vectors of $K,\tilde{K}$ by
  $\{\ket{v_i}\},\{\ket{\tilde{v}_i}\}$ respectively. These should
  not be confused with the underlying quantum states. Now
  $\ket{v_0}$ is the lowest right singular vector of $K$ with a
  corresponding zero singular value that corresponds to the exact
  Gibbs Hamiltonian. Up to an overall factor, the GH that we
  reconstruct from $\tilde{K}$ using its lowest $l$
  singular vectors corresponds to the projection of $\ket{v_0}$ on
  the subspace spanned by $\{\ket{\tilde{v}_i}\}_{i=1}^l$:
  \begin{align}
    \ket{\tilde{c}_l} \EqDef
      \sum_{i=0}^{l-1} \braket{\tilde{v}_i}{v_0}\ket{\tilde{v}_i}.
  \end{align}
  Therefore, $\ket{v_0} -\ket{\tilde{c}_l} = \sum_{i\ge
  l}\braket{\tilde{v}_i}{v_0}\ket{\tilde{v}_i}$ and so
  \begin{align}
  \label{eq:delta-v}
    \norm{\ket{v_0} -\ket{\tilde{c}_l}}^2
      = \sum_{i\ge l}|\braket{\tilde{v}_i}{v_0}|^2.
  \end{align}
  We estimate the above sum using first-order perturbation theory.
  By definition, $\tilde{K}=K+\epsilon E$, and therefore,
  \begin{align}
    \tilde{K}^T\tilde{K} &= (K+\epsilon E)^T(K+\epsilon E)\nonumber\\
      &= K^T K+\epsilon(E^T K + K^T E)+\epsilon^2 E^T E.
  \end{align}
  Treating $K^T K$ as an unperturbated Hamiltonian and $\epsilon(
  E^T K + K^T E + \epsilon E^T E)$ as a perturbation, we use
  first-order perturbation theory to estimate
  $\braket{\tilde{v}_i}{v_0}$. We note that if $(\ket{v_i},
  \lambda_i)$ are a pair of right singular-vector and singular value
  of $K$, then $(\ket{v_i}, \lambda_i^2)$ are pair of eigenvector
  and eigenvalue of $K^T K$. This allows us to apply standard
  perturbation theory calculation\cc{schrodinger1926quantization}
  for $\ket{\tilde{v}_i}$:
  \begin{align}
  \label{eq:perturbation_clean}
    \ket{\tilde{v}_i} =\ket{v_i}
      - \epsilon\sum_{j\neq i}\ket{v_j}
      \frac{\bra{v_j}E^T K + K^T E \ket{v_i}}
          {\lambda_j^2-\lambda^2_i} + \orderof{\epsilon^2}.
  \end{align}
  Multiplying by $\bra{v_0}$ and using the fact that $\lambda_0=0$,
  for any $i>0$ we get
  \begin{align*}
    \braket{v_0}{\tilde{v}_i}
      &= \epsilon \frac{\bra{v_0}E^T K
        + K^T E \ket{v_i}}{\lambda^2_i} + \orderof{\epsilon^2} \\
      &=\epsilon \frac{\bra{v_0}E^T\ket{u_i}}{\lambda_i}
        + \orderof{\epsilon^2},
  \end{align*}
  where $\ket{u_i}$ is the $i$'th left singular vector of $K$ and we
  used $\bra{v_0}K^T = 0$ and $K\ket{v_i} = \lambda_i\ket{u_i}$.
  Plugging this back to \Eq{eq:delta-v}, we get
  \begin{align*}
    \norm{ \ket{v_0}-\ket{\tilde{c}_l}}^2
      = \epsilon^2 \sum_{i\ge l}
        \frac{|\bra{v_0}E^T\ket{u_i}|^2}{\lambda^2_i}
          + O(\epsilon^3).
  \end{align*}
  To proceed, we would like to take the average of the above
  equation with respect to the i.i.d random variables that define
  the entries of $E$. Expanding $\ket{v_0}, \ket{u_i}$ in the
  standard basis, we get
  \begin{align*}
    \mathbb{E}&\big(|\bra{v_0}E^T\ket{u_i}|^2\big)
     = \mathbb{E}\big(\bra{v_0} E^T \ket{u_i}
       \bra{u_i}E\ket{v_0}\big) \\
     &= \sum_{k,l,m,n} \braket{v_0}{k}\braket{l}{u_i}
       \braket{u_i}{m}\braket{n}{v_0} \mathbb{E}\big(
         E_{lk}E_{mn}\big).
  \end{align*}
  Then using the fact that the entries of $E$ are i.i.d.\ normal
  random variables, we conclude that $\mathbb{E}\big(
  E_{lk}E_{mn}\big) = \delta_{lm}\cdot\delta_{kn}$, and so
  \begin{align*}
    \mathbb{E}\big(|\bra{v_0}E^T\ket{u_i}|^2\big)
     &= \sum_{k,l} |\braket{v_0}{k}|^2 \cdot
       |\braket{l}{u_i}|^2 \\
    &= \norm{\ket{v_0}}^2\cdot\norm{\ket{u_i}}^2=1.
  \end{align*}
  Overall then,
  \begin{align*}
    \mathbb{E}\big(\norm{ \ket{v_0}-\ket{\tilde{c}_l}}^2\big)
      = \epsilon^2 \sum_{i\ge l}\frac{1}{\lambda^2_i}
          + O(\epsilon^3) ,
  \end{align*}
  which proves the desired bound.
\end{proof}

Theorem~\ref{thm:perturb} uses first-order perturbation theory to
prove \Eq{eq:perturbation-bound-theorem}. Since to zero order the
singular values of $K$ and $\tilde{K}$ coincide,
$\lambda_i = \tilde{\lambda}_i + O(\epsilon)$, we can use the
empirical data to estimate the reconstruction error as
\begin{align}
\label{eq:perturb-empirical}
  \norm{\tilde{\bm{c}}_l-\bm{v}_0}
    \simeq \epsilon \left( \sum_{i\ge l}
      \frac{1}{\tilde{\lambda}^2_i}\right)^{1/2}
\end{align}
where $\epsilon = O(\frac{1}{\sqrt{m}})$, and $m$ is the number of
measurements used to estimate every entry of $\tilde{K}$. To verify
the applicability of this approximation, we performed numerical
simulations in which we calculated the LHS and the first-order RHS
of \Eq{eq:perturb-empirical}, and plotted the histogram of their
ratio for many random realizations. Specifically, for
each pair of values
$N\in\left\{5,6,7,8\right\},\epsilon\in\left\{ 10^{-2},10^{-3},10^{-4} \right\}$
we randomly picked ten normalized $2$-local
Hamiltonians on $N$ qubits and computed their CM $K$. Then, we added
to $K$ entries an i.i.d.\ Gaussian noise with amplitude $\epsilon$ to get
$\tilde{K}$, and calculated the projection of $\bm{v}_0$ on the
span of the lowest $l$ SVs of $\tilde{K}$,
for a wide range of $l$ values, in order to find $\bm{\tilde{c}}_l$.
The ratio between the LHS and the RHS of \Eq{eq:perturb-empirical}
was then calculated and averaged over all values of $N,\epsilon,l$.
The results are shown on \Fig{fig:error-estimation}
with an average ratio between LHS and RHS of $0.993$, which fits
\Eq{eq:perturb-empirical}.

\begin{figure}
    \label{fig:error-estimation} \centering
    \includegraphics[width=0.48\textwidth]{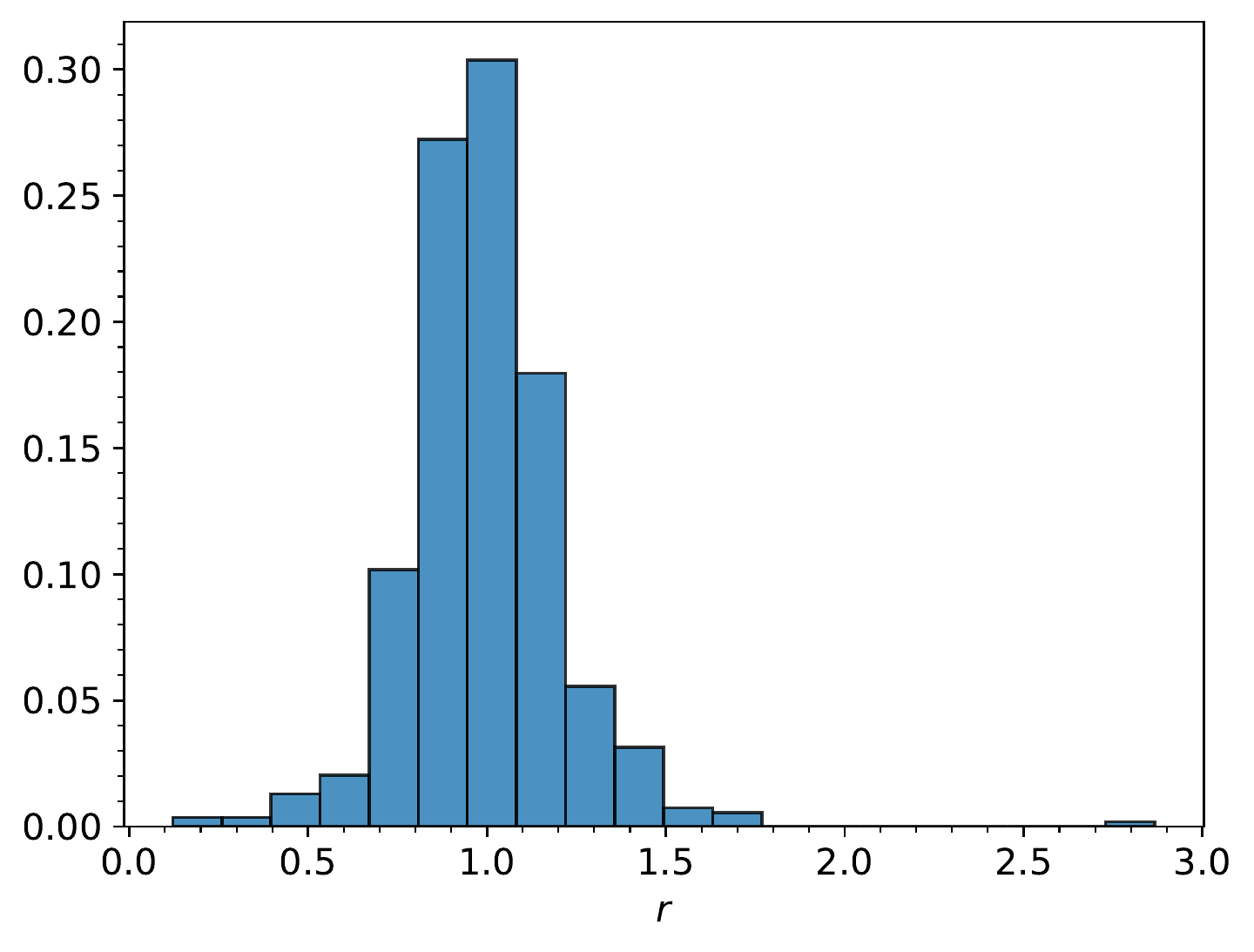}
    \caption{Histogram of the ratio between the
    reconstruction error for states with normalized random $2$-local GH
    and the error estimation from \Eq{eq:perturb-empirical}
    $r = \norm{\tilde{\bm{c}}_l-\bm{v}_0}
    / \left(\epsilon^2 \sum_{i\ge l}
      \frac{1}{\tilde{\lambda}^2_i}\right)^{1/2}$.
    System size was $N=5-8$ qubits and a gaussian noise with amplitude
    $\epsilon\in\left\{ 10^{-2},10^{-3},10^{-4} \right\}$
    was added to the CM to get $\tilde{K}$. Ten
    random normalized Hamiltonians were taken for each $N$ and $\epsilon$, and $r$
    was calculated for a wide range of $l$ values. The average $r$
    over $N,\epsilon,l$ was $0.993$.}
\end{figure}
\section{Gibbs State preparation with VarQITE}\label{sec:appendix-VarQITE}

Imaginary time evolution can be used for Gibbs state preparation.
Starting from maximally entangled state between sub-systems $A$ and
$B$, we evolve the system under the Hamiltonian $H_A\otimes I_B$ for
imaginary time $\beta/2$ and get the Gibbs state $e^{-\beta H}$ at
sub-system $A$.

In Refs.\cite{Yuan2019,Zoufal_2021} a variational algorithm for
imaginary time evolution, called VarQITE, is described. We followed
the variational circuit architecture described in\cc{Zoufal_2021}
and extended it from two qubits Gibbs state preparation up to
10-qubits, similar to \cRef{lehtonen_2021} (ring topology with $R_Y$
rotations). We ran the variational imaginary time evolution for the
transverse Ising model Hamiltonian (see \Eq{eq: transverse ising})
using 10 iterations. We used an imaginary time $1/2$ to get a final
state of $e^{-H}$. In order to avoid ill-conditioned linear
equations during the algorithm, we used ridge regularization, as
described in \cRef{Zoufal_2021}, with $\alpha=0.005$.  The algorithm
ran on the Qiskit simulator without noise. To assess the accuracy of
the algorithm (on simulations), we calculated the exact Gibbs state
using SciPy\cc{2020SciPy-NMeth} matrix exponentiation, and compared
it the output state of the noiseless variational circuit. The
fidelity between these two states is given in \Table{tab:gibbs
fidelity}.  Fidelities above $0.96$ were obtained for up to $6$
qubits and above $0.9$ up to $10$ qubits.

The actual states prepared on the quantum hardware differed
substantially from the intended Gibbs states (see
\App{sec:appendix-hardware-state-preparation}). Better results for
preparing the Gibbs state on the quantum hardware could probably be
obtained by performing the variational algorithm on the hardware,
instead of simulation, which would also account for noise in the
quantum hardware. However, running on quantum hardware would require
a large amount of quantum resources. For Gibbs state on $N$ qubits,
each algorithm iteration requires measuring $16N^2+4 N$ different
expectation values, each with different circuits. In addition, the
above circuits include $2N+1$ qubits and up to twice the number of
gates that are in the variational circuit. These long circuits would
have a large noise accumulation which will affect the algorithm
results. Moreover, if the optimization process takes too long,
drifting errors in the underlying quantum hardware might hinder the
quality of the final state. We leave this problem for future research.

\begin{table}
    \label{tab:gibbs fidelity}
    \centering
    \begin{tabular}{c|c}
        \toprule
        \textbf{\# of qubits} & \textbf{Fidelity} \\
        \midrule
        2                     & 0.9920            \\
        3                     & 0.9865            \\
        4                     & 0.9753            \\
        5                     & 0.9663            \\
        6                     & 0.9655            \\
        7                     & 0.9491            \\
        8                     & 0.9397            \\
        9                     & 0.9395            \\
        10                    & 0.9186            \\
        \bottomrule
    \end{tabular}
    \caption{Fidelity between exact 1D Ising model (\Eq{eq:
    transverse ising}) Gibbs state and the final state of the
    variational Gibbs state preparation algorithm, without noise. 10
    time iterations were made on simulation with a total (imaginary)
    time of $1/2$. The variational algorithm used ring architecture
    with parametrized $R_Y$ rotations before and after the ring (see
    \cRef{lehtonen_2021} for details).}
\end{table}

\section{Overlapping Local Tomography}
\label{sec:overlapping-local-tomography}

As described in \Sec{subsec:constraint-matrix}, the expectation
values that make up the entries of the constraint matrix are the
commutators of $k$-local Paulis (Hamiltonians basis $\left\{ S_m
\right\}$) with $k+1$-local Paulis (constraints, $\left\{ A_q
\right\}$). Therefore, these are the expectation values of
$2k$-local Paulis. In order to measure these expectation values
efficiently, we used a cyclic basis measurements protocol called
\emph{Overlapping Local Tomography}, which is describe in
\cRef{zubida2021optimized}. In our case, the $1D$ chain is divided
into cells of length $2k$, and each cell is measured using all
$3^{2k}$ measurement basis configurations. The measurement bases
are the same for each of the cells (see \Fig{fig:overlapping}).

\cRef{zubida2021optimized} also defines an extension of the protocol
for $D$ dimensional lattices, which can be used to apply the HLT
algorithm in higher dimensions.

\begin{figure}
    \label{fig:overlapping}
    \centering
    \includegraphics[width=0.3\textwidth]{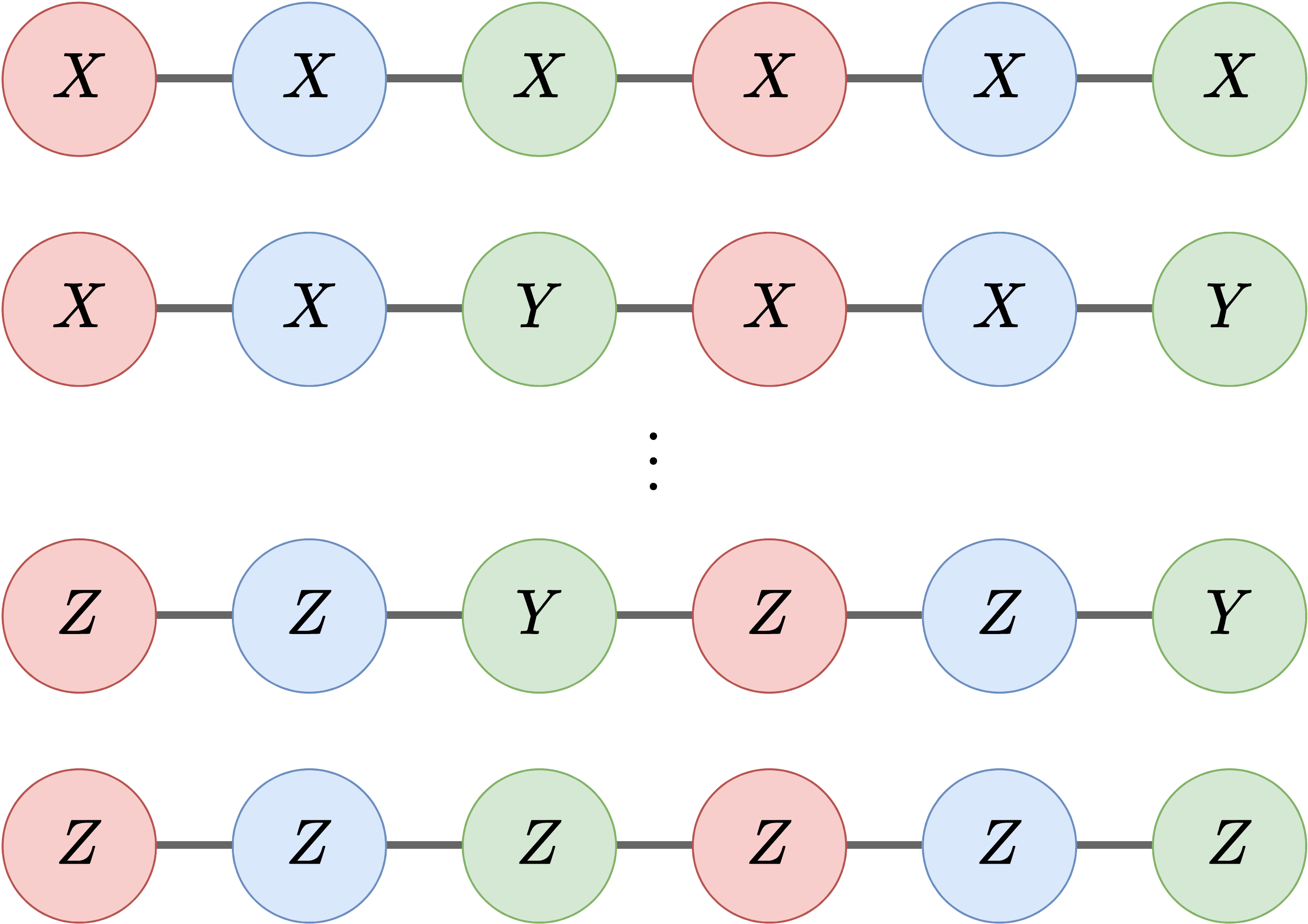}
    \caption{Visual description of \emph{Overlapping local tomography}
    measurement configurations for 1D chain with a unit length of $3$.
    Each consecutive $3$ qubits (different colors) form a unit cell which is
    measured in each of the $3^3$ basis configurations.}
\end{figure}

\section{Qiskit Quantum State Tomography}
\label{sec:appendix-QST}
\begin{table}[t]
    \label{tab:qst-fidelity}
    \centering
    \begin{tabular}{c|c|c|c}
        \toprule
        & \multicolumn{3}{c}{\textbf{Average Fidelity}} \\
        \cline{2-4}
        & \multicolumn{2}{c|}{\textbf{Gibbs}}    & \textbf{GHZ} \\
        \cline{1-4}
        & $m=8192$ & $m=20000$    & $m=8192$ \\
        \midrule
        N=2 & \quad 0.9994 \quad&\quad 0.9998  \quad &\quad 0.9975\quad       \\
        N=3 & \quad 0.9965 \quad&\quad 0.9987  \quad &\quad 0.9957\quad       \\
        N=4 & \quad 0.9882 \quad&\quad 0.9937  \quad &\quad 0.9952\quad       \\
        N=5 & \quad 0.9742 \quad&\quad 0.9853  \quad &\quad 0.9949\quad       \\
        \bottomrule
    \end{tabular}
    \caption{Average fidelity (over $10$ simulation runs) of Qiskit QST on
    variational Gibbs circuits for Transverse Ising
        (\Eq{eq: transverse ising}) and GHZ states.
        QST reconstructed a state on $N$ qubits with $m$ measurements in each
    of the $3^N$ circuits.}
\end{table}

This appendix summarizes the Qiskit QST algorithm, as described in
\cRef{QiskitQST}. Qiskit performs QST on $N$ qubits by measuring a
state $\rho$ in all of the $3^N$ Pauli bases. Each basis measurement
outcome can be represented as a projector $E_i$. Denoting $\vec{B}$
as the column vectorization of an operator $B$ (stacking $B$ columns
into a single column) and $p_i$ as the measured probabilities for
outcome $E_i$ gives the following linear equation:
\begin{align*}
E\vec{\rho}\EqDef\left(\begin{array}{c}
                             \vec{E}_{1}^{\dagger}\vec{\rho} \\
                             \vdots                          \\
                             \vec{E}_{2\cdot3^{N}}^{\dagger}\vec{\rho}
\end{array}\right)=\left(\begin{array}{c}
                             p_{1}  \\
                             \vdots \\
                             p_{2\cdot3^{N}}
\end{array}\right)=\vec{p}.
\end{align*}
This can be viewed as a least squares problem , whose solution is
the pseudo-inverse $\vec{\rho}=\left( E^T E \right)^{-1}E^T\vec{p}$~.
Due to finite statistics and quantum hardware noise, which affect
$\vec{p}$, the above solution may be non-physical, i.e.\ $\rho$
might contain negative eigenvalues. Qiskit addresses this issue by
using the algorithm described in \cRef{smolin2012efficient} which
computes the maximum-likelihood physical state for the results.

To assess the performance of Qiskit QST, which was used to verify the
HLT results on quantum hardware, we ran it on numerical simulations.
The simulations used the same variational circuits that created the
Gibbs states in \Sec{subsec:Variational Gibbs State}, and the same
number of measurements used on the quantum devices, i.e., $3^N$
different bases with $8192$ or $20000$ measurements each. We ran
each simulation for $10$ times and calculated the average fidelity
between the QST state and the exact state of the simulation. Results
are shown in \Table{tab:qst-fidelity}.  Evidently, for systems
of up to $5$ qubits, the average fidelity was above $0.985$ with
$20000$ measurements per basis. We also simulated Qiskit QST
on GHZ states and obtained fidelities above $0.99$ for up to $5$ qubits.

These results suggest that we can trust the fidelity estimates of
the HLT states with respect to the QST states on 5 qubits, as long
as the fidelity is below 0.98. Higher fidelities might not
necessarily represent the true fidelity of the HLT state with the
underlying quantum state.

\section{Quantum Hardware State Preparation}
\label{sec:appendix-hardware-state-preparation}
\begin{figure}
    \centering
    \includegraphics[width=0.5\textwidth]{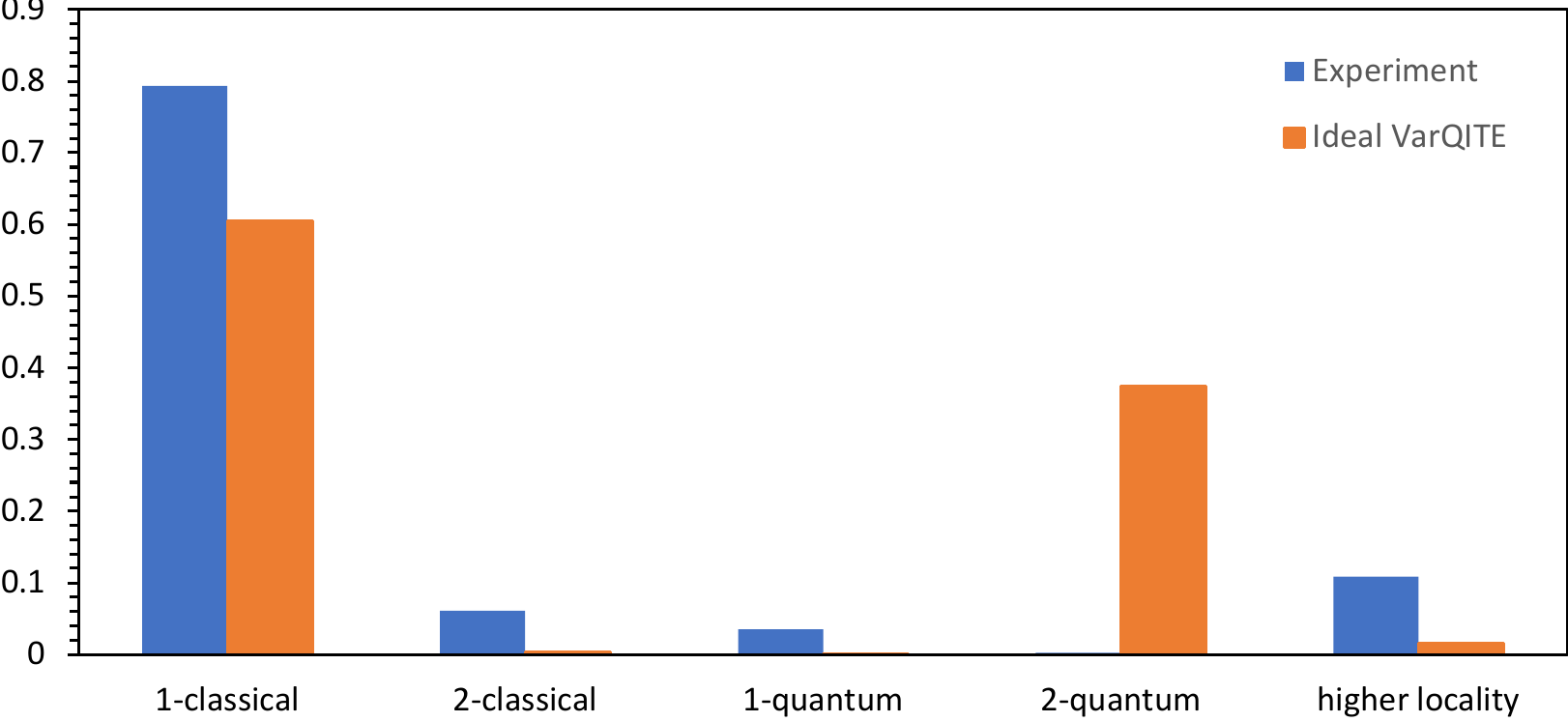}
    \\
    \vspace{10pt}
    \small
    \begin{tabular}{c|c|c|c}
        \toprule

        \text{5Q Experiment} & \text{}     &   \text{5Q Ideal VarQITE} & \text{}          \\
        \midrule
        IIIIZ       & \phantom{-}1.919 & IIIIZ       & 4.87185 \\
        IIZII       & \phantom{-}1.525 & ZIIII       & 4.85461 \\
        IZIII       & \phantom{-}1.150 & IIIZI       & 4.22767 \\
        IIIZI       & \phantom{-}0.472 & IZIII       & 4.14286 \\
        IYIII       & -0.232           & IIZII       & 4.13532 \\
        IIXII       & -0.234           & IXXII       & 4.11747 \\
        IIYII       & -0.240           & IIXXI       & 4.02958 \\
        IIIYI       & -0.242           & XXIII       & 3.75616 \\
        IXIII       & -0.263           & IIIXX       & 3.68091 \\
        IIIZZ       & -0.736           & IXZXI       & 0.92039 \\
        \bottomrule
    \end{tabular}
    \caption{Pauli decomposition of Gibbs Hamiltonians. GHs extracted
    from the QST-reconstructed state with $2\times 10^4$ measurements in each
    of the $3^5$ bases shown at
    Fig.~\hyperref[fig:hw-results-gibbs-cutoff]{4a} (blue), versus the
    ideal VarQITE circuit outcome without noise (orange).
    \textbf{Top}: $\left\| H
    \right\|_2^{-2} \sum_{p_i}\left( \Tr \left( H p_i
    \right)\right)^2$ where the sum is over different normalized Paulis
    for each column: 1-classical are single site $Z$, 2-classical are
    $ZZ$ on nearest neighbours, 1-quantum are single site $X$ or
    $Y$, 2-quantum are $2$-local which are not 2-classical, and
    higher locality are all the rest.
    \textbf{Bottom}: ten
    coefficients with the largest absolute values in the full Pauli
    decomposition (any range) of the GHs. }
    \label{fig:gibbs-preparation-a}
\end{figure}
\begin{figure}
    \centering
    \includegraphics[width=0.5\textwidth]{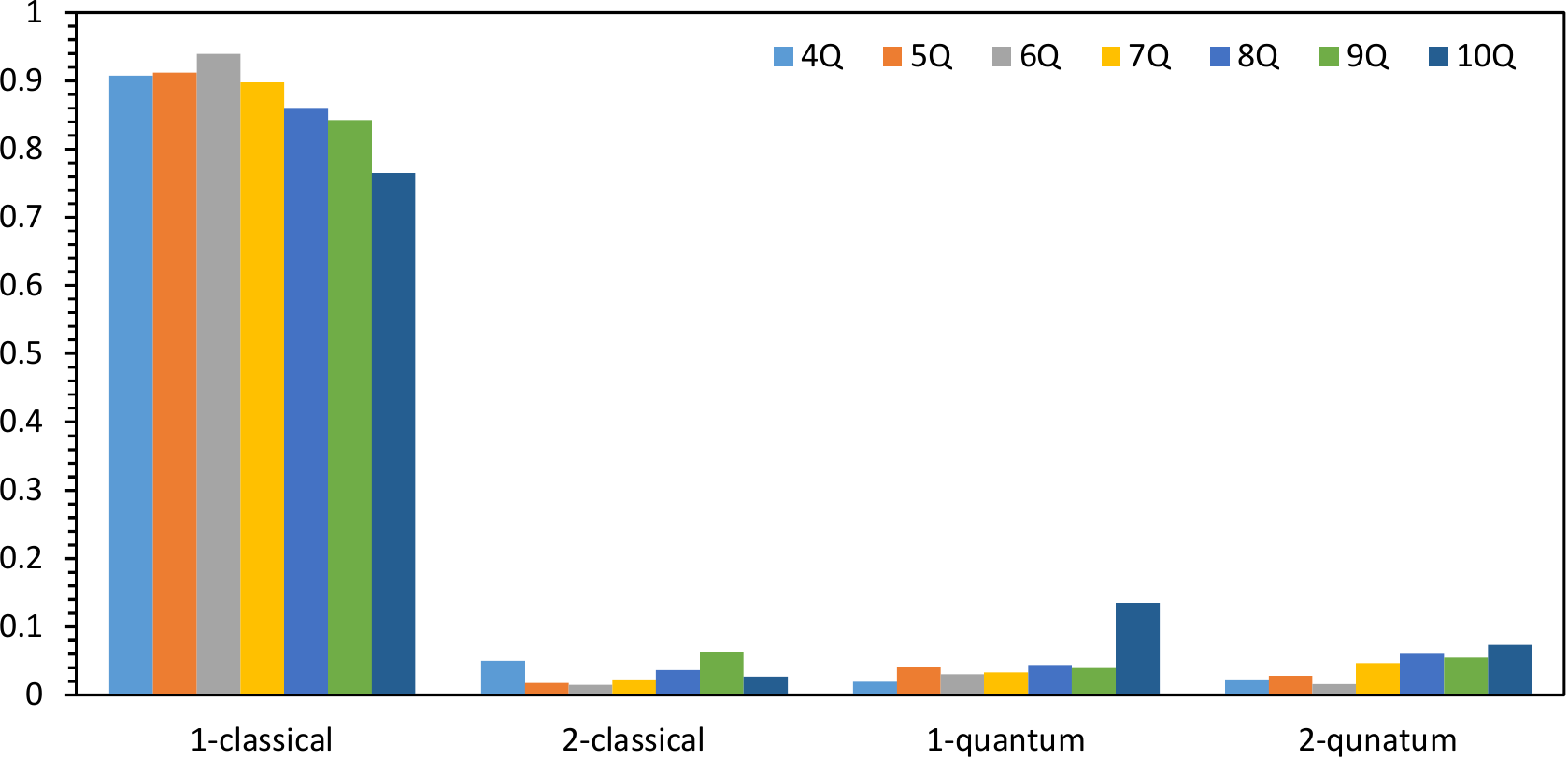}
    \\
    \centering
    \small
    \begin{tabular}{c|c|c|c|c|c|c|c|c|c|c|c|c|c}
          \toprule
          \textbf{4Q} & \textbf{}          & \textbf{5Q} & \textbf{}          & \textbf{6Q} & \textbf{}          & \textbf{7Q} & \textbf{}          \\
          \midrule
          IIIZ        & \phantom{-}1.389 & IIIIZ       & \phantom{-}1.869 & IIIIIZ      & \phantom{-}3.428 & IIIZIII     & \phantom{-}3.615 \\
          IIZI        & \phantom{-}0.982 & IIIZI       & \phantom{-}1.696 & IIZIII      & \phantom{-}2.337 & IIIIIIZ     & \phantom{-}3.255 \\
          ZIII        & \phantom{-}0.977 & IZIII       & \phantom{-}1.655 & IIIZII      & \phantom{-}2.300 & IIIIZII     & \phantom{-}3.031 \\
          IZII        & \phantom{-}0.471 & ZIIII       & \phantom{-}1.535 & ZIIIII      & \phantom{-}1.621 & IIIIIZI     & \phantom{-}2.864 \\
          IIYZ        & \phantom{-}0.141 & IIZII       & \phantom{-}1.521 & IIIIZI      & \phantom{-}1.083 & IIZIIII     & \phantom{-}2.469 \\
          XXII        & -0.122           & IIYII       & \phantom{-}0.522 & IZIIII      & \phantom{-}0.911 & IZIIIII     & \phantom{-}2.445 \\
          ZZII        & -0.149           & IIIIY       & \phantom{-}0.420 & IIXIII      & -0.353           & ZIIIIII     & \phantom{-}2.194 \\
          IZZI        & -0.155           & YZIII       & \phantom{-}0.300 & IIIIZZ      & -0.397           & IIIIIIY     & \phantom{-}0.847 \\
          IYII        & -0.224           & IIZZI       & -0.306           & XIIIII      & -0.453           & IXXIIII     & \phantom{-}0.702 \\
          IIZZ        & -0.425           & IZZII       & -0.353           & IIIIYI      & -0.485           & IIIIIZZ     & -1.015           \\
          \bottomrule
        \end{tabular}
    \\
    \begin{tabular}{c|c|c|c|c|c}
            \toprule
            \textbf{8Q} & \textbf{}         & \textbf{9Q} & \textbf{}         & \textbf{10Q} & \textbf{}         \\
            \midrule
            IIIIIIIZ    & \phantom{-}4.487 & IIIIIIIIZ   & \phantom{-}5.918 & IIIIIIIIIZ   & \phantom{-}10.410 \\
            IIIIIZII    & \phantom{-}4.450 & IIIIIIZII   & \phantom{-}5.581 & IIIIIIIZII   & \phantom{-}10.337 \\
            IIIZIIII    & \phantom{-}4.014 & IIIIZIIII   & \phantom{-}5.113 & IIIZIIIIII   & \phantom{0}\phantom{-}8.202  \\
            IIIIZIII    & \phantom{-}3.580 & IIIIIZIII   & \phantom{-}5.020 & IIIIIZIIII   & \phantom{0}\phantom{-}7.807  \\
            ZIIIIIII    & \phantom{-}3.559 & IIZIIIIII   & \phantom{-}3.782 & IIZIIIIIII   & \phantom{0}\phantom{-}6.035  \\
            IIIIIIZI    & \phantom{-}3.337 & IIIZIIIII   & \phantom{-}2.612 & IIIIIIIIZI   & \phantom{0}\phantom{-}5.465  \\
            IIZIIIII    & \phantom{-}2.690 & IZIIIIIII   & \phantom{-}2.515 & ZIIIIIIIII   & \phantom{0}\phantom{-}5.308  \\
            IZIIIIII    & \phantom{-}1.208 & IIIIIIIZI   & \phantom{-}2.095 & IIIIIIIZZI   & \phantom{0}-3.054            \\
            XIIIIIII    & -1.279           & ZIIIIIIII   & \phantom{-}1.721 & YIIIIIIIII   & \phantom{0}-4.206            \\
            ZZIIIIII    & -1.332           & IIIIIIIZZ   & -2.622           & XIIIIIIIII   & \phantom{0}-5.147            \\
            \bottomrule
        \end{tabular}
    \caption{Pauli decomposition of Gibbs Hamiltonians.
    GHs extracted from HLT results shown at
    \Fig{fig:partial-verification-results}.  \textbf{Top}: $\left\| H
    \right\|_2^{-2} \sum_{p_i}\left( \Tr \left( H p_i
    \right)\right)^2$ where the sum is over different normalized Paulis
    for each column: 1-classical are single site $Z$, 2-classical are
    $ZZ$ on nearest neighbours, 1-quantum are single site $X$ or
    $Y$, 2-quantum are $2$-local which are not 2-classical.
    HLT results with $k=2$ by definition does not have higher locality therefore
    it is not shown.
    \textbf{Bottom}: ten
    coefficients with the largest absolute values in the full Pauli
    decomposition (any range) of the GHs.}
    \label{fig:gibbs-preparation-b}
\end{figure}
\begin{figure}[b!]
    \centering
    \includegraphics[width=0.5\textwidth]{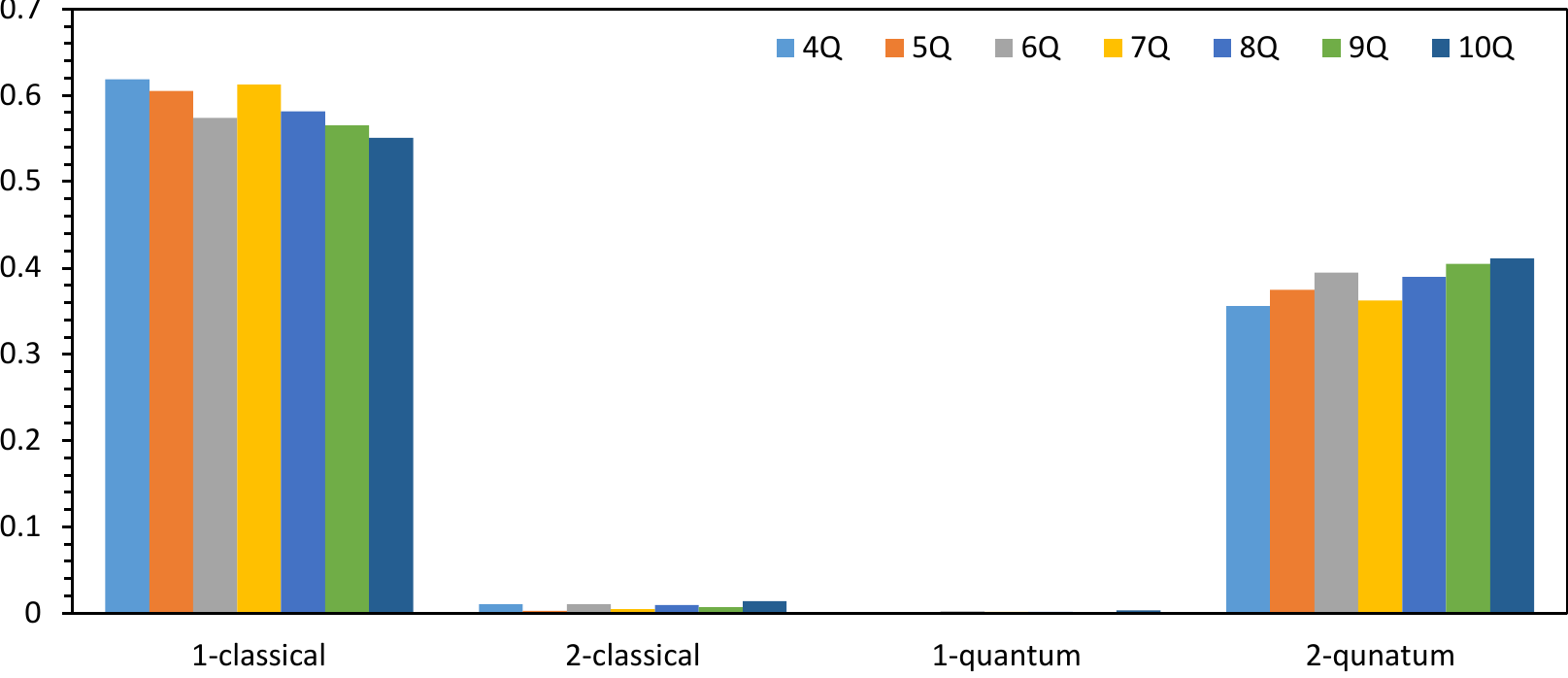}
    \\
    \centering
    \small
    \begin{tabular}{c|c|c|c|c|c|c|c|c|c|c|c|c|c}
          \toprule
          \textbf{4Q} & \textbf{}          & \textbf{5Q} & \textbf{}          & \textbf{6Q} & \textbf{}          & \textbf{7Q} & \textbf{}          \\
          \midrule
          ZIII        & \phantom{-}3.390 & IIIIZ       & 4.87185          & ZIIIII      & 6.915            & IIIIIIZ     & 10.480                      \\
          IIIZ        & \phantom{-}3.115 & ZIIII       & 4.85461          & IIIIIZ      & 6.765            & ZIIIIII     & 10.068                      \\
          IIZI        & \phantom{-}3.029 & IIIZI       & 4.22767          & IIZIII      & 6.277            & IIZIIII     & \phantom{1}9.325            \\
          IXXI        & \phantom{-}2.903 & IZIII       & 4.14286          & IIIIZI      & 6.163            & IIIIIZI     & \phantom{1}9.139            \\
          IZII        & \phantom{-}2.903 & IIZII       & 4.13532          & IIIXXI      & 6.039            & IZIIIII     & \phantom{1}8.651            \\
          XXII        & \phantom{-}2.687 & IXXII       & 4.11747          & IXXIII      & 5.844            & IIIIZII     & \phantom{1}8.616            \\
          IIXX        & \phantom{-}2.531 & IIXXI       & 4.02958          & IZIIII      & 5.795            & IIIXXII     & \phantom{1}8.438            \\
          IXZX        & \phantom{-}0.617 & XXIII       & 3.75616          & IIIZII      & 5.771            & IIIZIII     & \phantom{1}8.205            \\
          XZXI        & \phantom{-}0.607 & IIIXX       & 3.68091          & IIXXII      & 5.771            & IIIIIXX     & \phantom{1}7.852            \\
          IIZZ        & -0.526           & IXZXI       & 0.92039          & IIIIXX      & 5.393            & XXIIIII     & \phantom{1}7.772            \\
          \bottomrule
        \end{tabular}
    \\
    \begin{tabular}{c|c|c|c|c|c}
            \toprule
            \textbf{8Q} & \textbf{}         & \textbf{9Q} & \textbf{}         & \textbf{10Q} & \textbf{}         \\
            \midrule
            ZIIIIIII    & 14.527            & ZIIIIIIII   & 19.000            & IIIIIIIIIZ   & 27.529  \\
            IIIIIIIZ    & 13.417            & IIIIIIIIZ   & 18.212            & ZIIIIIIIII   & 26.464  \\
            IIIIIIZI    & 12.888            & IIIIIIZII   & 18.182            & IIIIIZIIII   & 25.502  \\
            IIIZIIII    & 12.319            & IIIIZIIII   & 18.003            & IIIZIIIIII   & 24.078  \\
            IIZIIIII    & 12.261            & IZIIIIIII   & 17.756            & IIIIZIIIII   & 24.029  \\
            XXIIIIII    & 12.061            & IIIIIZIII   & 17.471            & IIIXXIIIII   & 23.696  \\
            IIIXXIII    & 11.377            & IIZIIIIII   & 17.432            & IIIIIIZIII   & 23.303  \\
            IIIIIZII    & 11.368            & IIIZIIIII   & 17.265            & IIZIIIIIII   & 23.282  \\
            IZIIIIII    & 11.275            & IIIIIXXII   & 17.163            & IXXIIIIIII   & 22.232  \\
            IIIIXXII    & 11.014            & IIIIXXIII   & 17.015            & IIIIIIIIXX   & 22.227  \\
            \bottomrule
        \end{tabular}
    \caption{Pauli decomposition of Gibbs Hamiltonians.
    GHs extracted from Ising Model VarQITE circuits' ideal
    output --- without hardware noise
    (see \App{sec:appendix-VarQITE},
    and specifically \Table{tab:gibbs fidelity}).
    \textbf{Top}: $\left\| H
    \right\|_2^{-2} \sum_{p_i}\left( \Tr \left( H p_i
    \right)\right)^2$ where the sum is over different normalized Paulis
    for each column: 1-classical are single site $Z$, 2-classical are
    $ZZ$ on nearest neighbours, 1-quantum are single site $X$ or
    $Y$, 2-quantum are $2$-local which are not 2-classical.
    \textbf{Bottom}: ten
    coefficients with the largest absolute values in the full Paulis
    decomposition (any range) of the GHs.}
    \label{fig:gibbs-preparation-c}
\end{figure}
Noise within quantum hardware\cc{clerk2010introduction} causes
imperfect state preparation. Therefore, the Gibbs states prepared
by the algorithm from \App{sec:appendix-VarQITE} were not
exact.

We extracted the actual Gibbs Hamiltonians
prepared by the quantum hardware
using the density matrix from QST results and the matrix log
function of SciPy\cc{2020SciPy-NMeth}.
We calculated the Pauli decomposition of these Hamiltonians using
the trace inner product $\left< A,B \right>\EqDef\Tr\left( A^\dagger B
\right)$. The 10 coefficients with the largest absolute values in
the decompositions of the results from
Figs.~\hyperref[fig:hw-results-gibbs-cutoff]{4a},~\ref{fig:partial-verification-results}
are shown at
Figs.~\hyperref[fig:gibbs-preparation-a]{9 (left)},\ref{fig:gibbs-preparation-b}
respectively, and can be compared to the values of the ideal
VarQITE circuit, without hardware noise, in
Figs.~\hyperref[fig:gibbs-preparation-a]{9 (right)},\ref{fig:gibbs-preparation-c}
respectively.
It is noted that the classical part of the
Pauli decomposition (1-classical columns in
\Figss{fig:gibbs-preparation-a}{fig:gibbs-preparation-b}{fig:gibbs-preparation-c})
of the original Gibbs Hamiltonians was moderately changed
by the noise on the quantum hardware.
On the other hand,
the $XX$ terms did not remain as in the original Hamiltonians, and
their coefficients have changed drastically. Nevertheless, the
prepared Gibbs state was not entirely classical, as it still contained
non-negligible non-commuting terms.  Furthermore, the
Gibbs Hamiltonians remained almost $2$-local (as the original Gibbs
Hamiltonians), which explains the success of the HLT method.

GHZ states were also subject to noise during their preparations.
However, in this case, the noise led to highly non-local Gibbs
Hamiltonians, which is reflected in the non-locality of their Pauli
decomposition. Indeed, in the case of 5 qubits (obtained after
preparing a 6-qubit GHZ state and tracing over one qubit), the resultant
Gibbs Hamiltonian had a lot of weight on 3-local and higher terms.
Summing over all the contributions of the 2-local, nearest-neighbor
Pauli operators, we obtained
\begin{align*}
  \left\| H \right\|_2^{-2} \sum_{\text{2-local $P_s$}}\left( \Tr \left( H
    P_s  \right)\right)^2 < 0.6.
\end{align*}
In other words, over $40\%$ of the weight of the Gibbs Hamiltonian was
in non-local terms.
\end{document}